\documentclass[screen, sigplan]{acmart}
\DeclareMathAlphabet{\altmathcal}{OMS}{cmsy}{m}{n}

\usepackage{amsmath,amsfonts,amssymb}
\usepackage{algpseudocode}
\usepackage{algorithm}
\usepackage{graphicx}
\usepackage{textcomp}
\usepackage{xcolor}
\usepackage{adjustbox}
\usepackage[htt]{hyphenat} 
\usepackage{semantic} 
\usepackage{soul}
\usepackage{xspace}
\usepackage[frozencache,cachedir=.]{minted}
\usepackage{pgfplots}
\usepackage{pgfplotstable}
\pgfplotsset{compat=1.17}
\usepackage{catchfile}
\usepackage{ragged2e}
\usepackage{caption}
\usepackage{subcaption} 
\usepackage[utf8]{inputenc}
\usepackage{tabularx}
\usepackage{listings}
\usepackage{amsmath}
\usepackage{xspace}
\usepackage{float}
\usepackage{microtype}
\usepackage[T1]{fontenc}
\usepackage{booktabs}
\usepackage{color}
\usepackage{enumerate}
\usepackage{multirow}
\usepackage{syntax}
\usepackage[referable]{threeparttablex}
\usepackage[title]{appendix}
\usepackage{stmaryrd} 
\usepackage{siunitx} 
\usepackage{enumitem}
\usepackage{caption}
\usepackage{arydshln}
\usepackage{dashrule}

\usepackage{tikz-cd}
\usepackage{tikz}

\usepackage{mathtools}
\usepackage{mathpartir}

\usepackage[capitalize,noabbrev]{cleveref}

\newcommand{\koika}{Kôika\xspace}

\newcommand{\lr}{Lakeroad\xspace}

\newcommand{\lrfn}{\text{$f_{\textsc{lr}}$}\xspace}   
\newcommand{\lrfnbmc}{\text{$f_{\textsc{lr}}^{*}$}\xspace}   

\newcommand{\SynProg}{\mathsf{Prog}\xspace}
\newcommand{\SynId}{\mathsf{Id}\xspace}
\newcommand{\SynBv}{\mathsf{BV}\xspace}
\newcommand{\SynVar}{\mathsf{Var}\xspace}
\newcommand{\SynNode}{\mathsf{Node}\xspace}

\newcommand{\Op}{\mathsf{OP}\xspace}
\newcommand{\OpBv}{\Op_{bv}}
\newcommand{\OpWire}{\Op_{w}}
\newcommand{\IRReg}{\lstinline[language=thelang]{Reg}\xspace}
\newcommand{\IRPrim}{\lstinline[language=thelang]{Prim}\xspace}

\newcommand{\Time}{\textsf{Time}\xspace}

\newcommand{\Sketch}{\textsc{Sketch}\xspace}

\definecolor{navy}{HTML}{0f1566}

\newtheorem{property}{Property}
\definecolor{mygreen}{rgb}{0,0.6,0}
\lstdefinestyle{lispstyle}{
  backgroundcolor=\color{white},
  basicstyle=\ttfamily\footnotesize,
  breakatwhitespace=false,
  breaklines=true,
  captionpos=b,
  commentstyle=\color{mygreen},
  extendedchars=true,
  keepspaces=true,
  keywordstyle=\color{black},
  language=Lisp,
  morekeywords={*,...},
  numbers=none,
  numbersep=5pt,
  numberstyle=\tiny\color{mygray},
  rulecolor=\color{black},
  showspaces=false,
  showstringspaces=false,
  showtabs=false,
  stringstyle=\color{black},
  tabsize=2,
  title=\lstname
}
\lstdefinestyle{pystyle}{
  backgroundcolor=\color{white},
  basicstyle=\ttfamily\footnotesize,
  breakatwhitespace=false,
  breaklines=true,
  captionpos=b,
  commentstyle=\color{mygreen},
  extendedchars=true,
  keepspaces=true,
  keywordstyle=\color{blue},
  language=Python,
  morekeywords={*,...},
  numbers=none,
  numbersep=5pt,
  numberstyle=\tiny\color{mygray},
  rulecolor=\color{black},
  showspaces=false,
  showstringspaces=false,
  showtabs=false,
  stringstyle=\color{black},
  tabsize=2,
  title=\lstname
}

\newcommand\YAMLcolonstyle{\color{black}\mdseries}
\newcommand\YAMLkeystyle{\color{black}\bfseries}
\newcommand\YAMLvaluestyle{\color{black}\mdseries}

\makeatletter

\newcommand\language@yaml{yaml}

\expandafter\expandafter\expandafter\lstdefinelanguage
\expandafter{\language@yaml}
{
  keywords={true,false,null,y,n},
  keywordstyle=\color{darkgray}\bfseries,
  basicstyle=\color{black}\linespread{0.8}\ttfamily\footnotesize,
  comment=[l]{\#},
  morecomment=[s]{/*}{*/},
  commentstyle=\color{purple}\ttfamily,
  stringstyle=\YAMLvaluestyle\ttfamily,
  moredelim=[l][\color{orange}]{\&},
  moredelim=[l][\color{magenta}]{*},
  moredelim=**[il][\YAMLcolonstyle{:}\YAMLvaluestyle]{:},   
  morestring=[b]',
  morestring=[b]",
  literate =    {---}{{\ProcessThreeDashes}}3
                {>}{{\textcolor{red}\textgreater}}1     
                {|}{{\textcolor{red}\textbar}}1 
                {\ -\ }{{\mdseries\ -\ }}3,
}

\lst@AddToHook{EveryLine}{\ifx\lst@language\language@yaml\YAMLkeystyle\fi}
\makeatother

\newcommand\ProcessThreeDashes{\llap{\color{cyan}\mdseries-{-}-}}

\DeclareMathOperator{\defn}{\Coloneqq}

\lstdefinelanguage{thelang}{
  basicstyle=\ttfamily,
  keywordstyle=\color{black}\bfseries,
  morekeywords=[1]{let,in,:=,Reg,Prim,Op},
  morekeywords=[2]{},
  morekeywords=[3]{},
  alsoletter={:=},
  morestring=[b]",
  morecomment=[l]{\#},
  morecomment=[s]{(*}{*)},
  moredelim=**[is][\color{white}]{(&}{&)},
}

\newcommand{\seq}[1]{\langle#1\rangle}

\newcommand{\UberLang}{\ensuremath{\altmathcal{L}_\textsc{lr}}\xspace}
\newcommand{\SpecLang}{\ensuremath{\altmathcal{L}_\textsc{beh}}\xspace}
\newcommand{\ImplLang}{\ensuremath{\altmathcal{L}_\textsc{struct}}\xspace}
\newcommand{\SketchLang}{\ensuremath{\altmathcal{L}_\textsc{sketch}}\xspace}

\newcommand{\tighten}{\looseness=-1}

\newcommand{\asplossubmissionnumber}{71}

\usepackage[normalem]{ulem}
\begin{abstract}
\begin{figure}
\centering

\includegraphics[width=0.91\columnwidth]{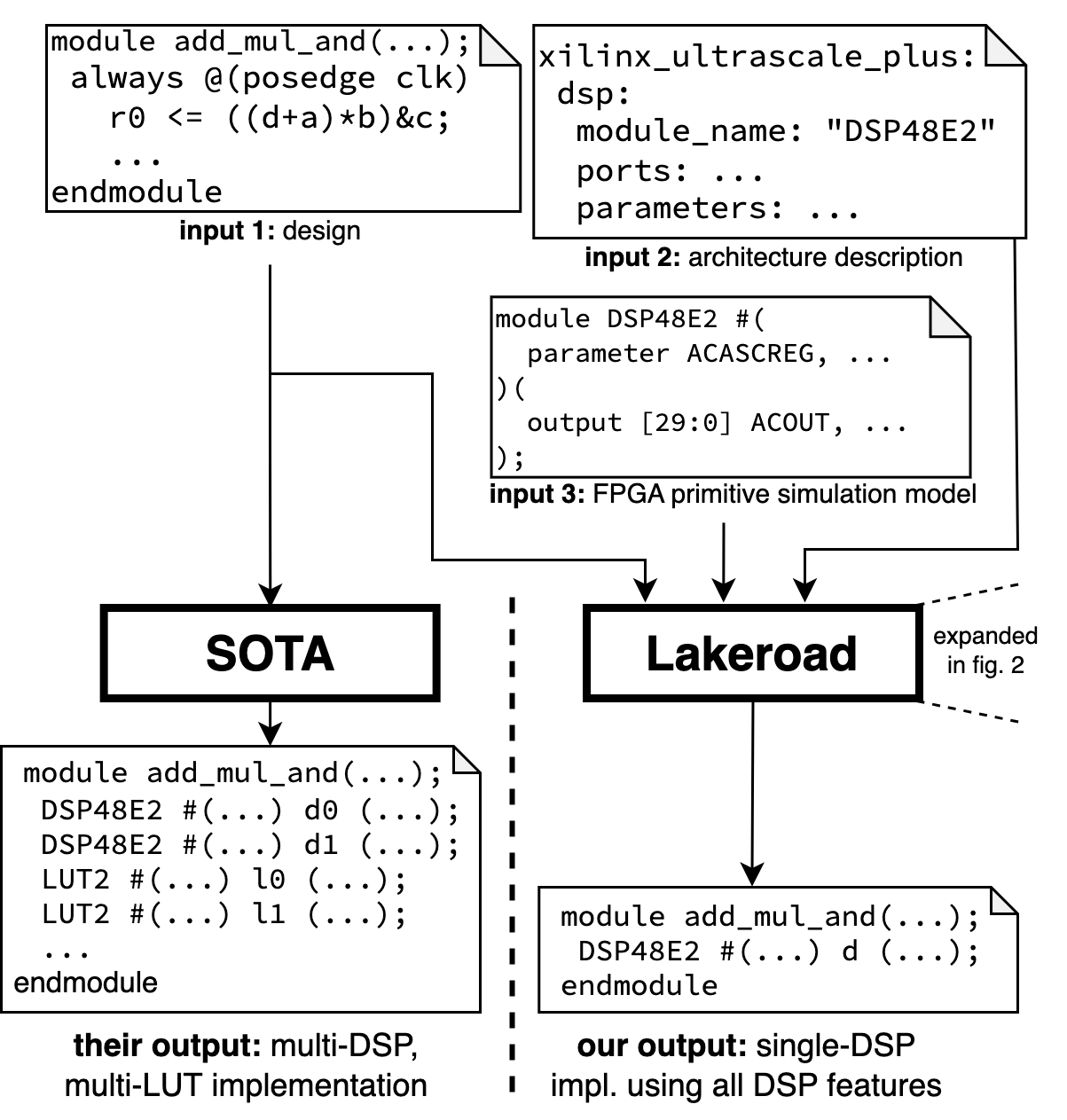}

\vspace{-3mm}
\caption{
Even given a simple input
  design (input 1),
  the state-of-the-art (SOTA)
  hardware synthesis tool
  for Xilinx FPGAs
  frequently
  fails to efficiently use 
  programmable primitives
  like DSPs.
\lr,
  on the other hand,
  can utilize all features
  of programmable primitives
  given just a short description
  of an FPGA architecture (input 2)
  and the vendor-provided 
  simulation models
  of the primitive (input 3).\tighten
}
\label{fig:firstpage}

\vspace{-5mm}
\end{figure}

FPGA technology mapping is the process of
  implementing a hardware design expressed in 
  high-level HDL (hardware design language) code
  using the low-level, architecture-specific primitives of 
  the target FPGA.
As FPGAs become increasingly heterogeneous, 
  achieving high performance
  requires hardware synthesis tools 
  that better support mapping to complex, 
  highly configurable primitives 
  like digital signal processors (DSPs).
Current tools
  support DSP mapping via handwritten special-case mapping rules,
  which are laborious to write, error-prone, and often overlook mapping opportunities.
We introduce \lr,
  a principled approach to technology mapping via
  sketch-guided program synthesis.
\lr leverages two techniques---architecture-independent 
  sketch templates and semantics extraction from HDL---to
  provide extensible technology mapping 
  with stronger correctness guarantees
  and higher coverage of 
  mapping opportunities
  than state-of-the-art tools.
Across representative microbenchmarks,
  \lr produces
  2--3.5$\times$ the number of optimal mappings
  compared to proprietary state-of-the-art tools
  and
  6--44$\times$ the number of optimal mappings
  compared to popular open-source tools,
  while also providing correctness guarantees
  not given by any other tool.

\end{abstract}

\copyrightyear{2024}
\acmYear{2024}
\setcopyright{rightsretained}
\acmConference[ASPLOS '24]{29th ACM International Conference on Architectural Support for Programming Languages and Operating Systems, Volume 2}{April 27-May 1, 2024}{La Jolla, CA, USA}
\acmBooktitle{29th ACM International Conference on Architectural Support for Programming Languages and Operating Systems, Volume 2 (ASPLOS '24), April 27-May 1, 2024, La Jolla, CA, USA}
\acmDOI{10.1145/3620665.3640387}
\acmISBN{979-8-4007-0385-0/24/04}

\begin{document}

\title[FPGA Technology Mapping Using Sketch-Guided Program Synthesis]{FPGA Technology Mapping Using \\ Sketch-Guided Program Synthesis}
\date{}
\author{Gus Henry Smith}
\orcid{0000-0001-9754-233X}
\affiliation{%
    \institution{University of Washington}
    \city{Seattle}
    \country{USA}}
\email{gussmith@cs.washington.edu}

\author[BK]{Ben Kushigian}
\orcid{0009-0009-2504-1582}
\affiliation{%
    \institution{University of Washington}
    \city{Seattle}
    \country{USA}}
\email{benku@cs.washington.edu}

\author[VC]{Vishal Canumalla}
\orcid{0009-0001-5418-1279}
\affiliation{%
    \institution{University of Washington}
    \city{Seattle}
    \country{USA}}
\email{vishalc@cs.washington.edu}

\author[AC]{Andrew Cheung}
\orcid{0009-0006-0661-2640}
\affiliation{%
    \institution{University of Washington}
    \city{Seattle}
    \country{USA}}
\email{acheung8@cs.washington.edu}

\author[SL]{Steven Lyubomirsky}
\orcid{0009-0003-6747-7014}
\affiliation{
  \institution{OctoAI}
  \city{Seattle}
  \country{USA}
}
\email{slyubomirsky@octo.ai}

\author[SP]{Sorawee Porncharoenwase}
\orcid{0000-0003-3900-5602}
\affiliation{
  \institution{University of Washington}
  \city{Seattle}
  \country{USA}
}
\email{sorawee@cs.washington.edu}

\author[RJ]{Ren{\'e} Just}
\orcid{0000-0002-5982-275X}
\affiliation{%
    \institution{University of Washington}
    \city{Seattle}
    \country{USA}}
\email{rjust@cs.washington.edu}

\author[GLB]{Gilbert Louis Bernstein}
\orcid{0000-0002-3016-1169}
\affiliation{%
    \institution{University of Washington}
    \city{Seattle}
    \country{USA}}
\email{gilbo@cs.washington.edu}

\author[ZT]{Zachary Tatlock}
\orcid{0000-0002-4731-0124}
\affiliation{%
    \institution{University of Washington}
    \city{Seattle}
    \country{USA}}
\email{ztatlock@cs.washington.edu}
\maketitle

\thispagestyle{empty}

\section{Introduction}
\label{sec:intro}

Given a high-level hardware design specification
  (e.g., expressed in behavioral Verilog),
  FPGA technology mappers
  search for an equivalent
  low-level implementation
  in terms of the target FPGA's
  primitives.
See \cref{fig:firstpage} for an example, where
the high-level, behavioral \texttt{add\_mul\_and}
  module (``input 1'')
  is converted into FPGA-specific implementations
  (``their output'' and ``our output'')
  using Xilinx-specific
  \texttt{DSP48E2} and \texttt{LUT2} primitives.
  
Historically,
  FPGAs consisted of relatively simple
  primitives, such as
  lookup tables (LUTs) and carry chains.
Tools like
  ABC~\cite{ABC,abc2,brayton2010abc}
  \textit{automatically} 
  map to these basic primitives
  by translating designs
  to a library of simple logic gates
  and then packing those gates
  into LUTs.

However, FPGAs are becoming
  increasingly heterogeneous
  via
  the inclusion of specialized and diverse primitives
  such as digital signal processors (DSPs).
Utilizing these specialized primitives
  effectively
  is now
  crucial for achieving
  high performance~\cite{vega2021reticle}.
These specialized primitives
  make FPGA technology mapping far more challenging
  since technology mappers must now
  explore a much larger search space
  while also satisfying each primitive's
  complex set of restrictions and dependencies.
For example, Xilinx's DSP48E2
  is a multifunction 
  DSP
  with nearly
  100 ports and parameters,
  whose numerous configurations
  enable 
  support for a large variety of computations.
The manual for the DSP48E2 alone
  is 75 pages long,
  where considerable text details
  the complex restrictions
  between the settings of the nearly 100
  ports and parameters.

Existing technology mapping tools
  frequently fail to map designs
  to
  specialized primitives like DSPs,
  requiring manual work for the hardware designer
  to recover the performance of their design.
While existing toolchains
  have the ability to automatically infer
  locations where specialized primitives
  can be used in large designs,
  inference often fails~\cite{xilinxforum1,xilinxforum2,inferringreddit}.
In these cases, the designer can either
  accept lower performance and higher resource
  utilization,
  or they can perform
  what we call
  \textit{partial design mapping.}
During partial design mapping, 
  the designer
  manually identifies and separates out
  the module that should be mapped
  to a DSP.
They can attempt to re-run technology mapping
  on that module alone,
  in the hopes that mapping succeeds.
Yet existing toolchains often fail
  \textit{even in the partial design mapping case:}
  \cref{fig:firstpage} shows a
  simple module
  \texttt{add\_mul\_and}
  which \textit{should} fit on a single DSP48E2
  according to the DSP's manual,
  but is instead mapped to multiple DSPs and LUTs
  by current state-of-the-art tools.\footnote{
    Licensing restrictions forbid naming the
    specific proprietary tools, but they are familiar,
    standard packages used by many hardware designers.
  }
In the worst case,
  hardware designers are forced to manually
  instantiate complex primitives by hand,
  e.g., by looking through the 75-page
  DSP48E2 user manual
  to manually configure the DSP's dozens
  of ports and parameters.

Current state-of-the-art
  technology mappers 
  are implemented via
  ad hoc, handwritten pattern matching procedures,
  which
  fall short in three primary ways.
First,
  as we saw above,
  they are \textbf{incomplete:}
  they miss many mapping opportunities,
  even across microbenchmarks based on vendor documentation.
Second, they \textbf{do not provide strong correctness guarantees:}
  recent work highlights the significant number of bugs found across 
  all major hardware synthesis tools~\cite{herklotz2020finding}.
Third, they are \textbf{difficult to extend:}
  \textit{each} new complex primitive requires
  supporting detailed semantics
  and adding hundreds of new, special-case
  syntactic pattern matching rules~\cite{wolf2013yosys}.

This paper's
  key observation is that 
  technology mapping
  is well-suited for the application
  of automated reasoning procedures---%
  specifically,
  \textit{program synthesis}~\cite{gulwani2017program}.
We observe that 
  the configuration space of
  a programmable FPGA primitive
  is essentially a small, bespoke
  programming language,
  and that
  program synthesis
  could be applied
  to automatically generate
  primitive configurations.
We explore how
  program synthesis
  can simplify the design and implementation of
  FPGA technology mappers while providing
  \textbf{correct},
  \textbf{extensible}, and
  \textbf{more complete}
  support for mapping to 
  diverse, highly configurable primitives
  like DSPs.
Program synthesis techniques rely on
  automated theorem provers like
  SAT and SMT solvers~\cite{de2008z3, barbosa22cvc5}
  to automatically generate programs
  satisfying a given specification.
We demonstrate how
 \textit{sketch-guided program synthesis}~\cite{solar2008program}
  can be adapted
  for FPGA technology mapping,
  leveraging the
  Rosette~\cite{torlak2014lightweight} 
  program synthesis framework.


Sketch-guided program synthesis requires
  encoding the \textit{semantics}
  of the target language:
  in our case,
  a machine-readable,
  mathematical model specifying
  the behavior of each
  FPGA-specific primitive
  being mapped to.
In a typical synthesis tool,
  which generates programs
  for a single target language,
  this is a one-time cost.
However,
  in our setting,
  each new FPGA primitive
  introduces yet another new target language,
  which in turn requires
  extending the tool to encode
  yet another formal semantics.

To support
  correct, extensible, and more complete
  technology mapping, we propose
  automating this process with
  \textbf{semantics extraction from HDL}, 
  adapted from past work~\cite{daly2022synthesizing},
  to automatically extract
  complete primitive semantics
  from vendor-published HDL models
  (\cref{fig:firstpage}, ``input 3'').
Traditionally, such models have
   been used only 
  for
  simulation and validation
  \textit{after} technology mapping;
  we show that using the semantics
  to
  \textit{implement}
  technology mapping
  with a program-synthesis-based approach
  yields substantially more
  complete FPGA technology mapping.

Sketch-guided program synthesis also
  requires \textit{sketches}, which are 
  partially complete programs with ``holes'' to be filled in
  by the solver.
Sketches primarily serve to
  scale synthesis by
  constraining the set of programs that 
  solvers explore when searching for
  one that satisfies
  the given specification,
  i.e., performance at the cost of completeness.
In our setting,
  sketches correspond to
  arrangements of primitives,
  using holes
  as placeholders
  for some of the primitives' 
  ports and parameters.
This could be
  a single DSP with holes for
  its ports and parameters
  (as in the example in \cref{sec:overview-part-2}),
  or a number of LUTs with holes for their LUT memories,
  or even a mixture of LUTs, DSPs, and carry chains.
The synthesizer ``fills in the holes''
  as necessary for
  the low-level FPGA-specific primitive to implement
  a given high-level behavioral design fragment.
Unfortunately,
  developing effective sketches
  still requires synthesis expertise~\cite{10.1145/3140587.3062353,vanGeffenJITSynth}.
  Na\"ively, our approach would also
  require new sketches for each new
  FPGA primitive we target.


To address these challenges, 
  we introduce
  \textbf{architecture-\newline independent sketch templates.}
Hardware designs are often implemented
  using high-level blueprints that are similar 
  across most FPGA architectures---%
  sketch templates
  capture these blueprints
  and make them reusable across architectures.
Therefore, by using sketch templates, we
  greatly reduce the overhead of supporting
  new architectures and
  diverse primitives.
Typically, when adding support for
  a new primitive or FPGA architecture in \lr,
  the hardware designer
  need not write or modify
  any sketch templates.


We leverage
  semantics extraction
  from HDL
  and architecture-independent
  sketch templates
  to build \lr,\footnote{
  \lr is publicly available at
  \url{https://github.com/uwsampl/lakeroad}.
  }
  a new FPGA technology mapper
  based on program synthesis.
  
\lr's prototype implementation automatically
  imports semantics for the LUTs, arithmetic carry chains,
  and DSPs of the Xilinx UltraScale+, Lattice ECP5,
  Intel Cyclone 10 LP,
  and SOFA~\cite{sofa} FPGA architectures. 
The only additional user input to \lr is a short 
  architecture description
  that lists the target FPGA's
  primitives (\cref{fig:firstpage}, ``input 2'').
Architecture descriptions
  only need to be written once per architecture,
  and \lr pre-supplies architecture descriptions
  for the aforementioned architectures.
With the automatically
  extracted primitive semantics
  and the user-provided architecture description,
  we demonstrate that \lr
  is more complete than proprietary tools on a variety
  of microbenchmarks
that are representative of program fragments
  implemented with DSPs during partial design mapping.
In particular,
  \lr maps up to 3.5$\times$ 
  more microbenchmarks than 
  state-of-the-art
  tools for Xilinx, Lattice, and Intel,
  and up to 44$\times$ 
  more microbenchmarks 
  than Yosys.\looseness=-1

This paper makes  
the following key contributions:
\begin{itemize}[leftmargin=*]
\item The novel application of program synthesis
  to produce a technology mapper---\lr---that is
  more \textbf{correct, complete,} and \textbf{extensible} than state-of-the-art 
  tools.
\item A technique for applying
  \textbf{semantics extraction from HDL}
  to automatically generate models
  of hardware usable by 
  formal automated reasoning tools.
\item The concept of 
  \textbf{architecture-independent sketch templates,}
  which capture common patterns in hardware design
  in an architecture-independent way,
  plus \textbf{primitive interfaces} and \textbf{architecture descriptions}, the abstractions
  underlying these templates.
\item A formalization of the \lr
  toolchain and an argument
  for its correctness and sketch-completeness.
\item The first notion of \textbf{technology
  mapping completeness} for FPGA
  compilers.
\item \textbf{Empirical comparisons }of
  \lr and existing hardware synthesis
  tools to evaluate both their
  relative completeness and
  ease of extensibility.



\end{itemize}

In the following sections, we walk through a real-world example
  using both existing tools
  and \lr
  and 
  highlight \lr's design and key features (\cref{sec:overview}); 
  formalize \lr and
  demonstrate its correctness (\cref{sec:formalization}); 
  describe \lr's implementation (\cref{sec:implementation}); and 
  evaluate \lr on its
  completeness of mapping,
  extensibility,
  and expressiveness (\cref{sec:evaluation}) .
\cref{sec:background-and-related-work}
  discusses related work, and
  \cref{sec:conclusion} concludes.
\section{Overview}
\label{sec:overview}

We now walk through an 
  example
  of how current FPGA technology mapping tools can fail
  a hardware designer (\cref{sec:overview-part-1}) 
  and how \lr overcomes these limitations (\cref{sec:overview-part-2}).
In the process,
  we provide a high-level overview
  of \lr's main components.

\subsection{Compiling a Design to a DSP with Existing Tools}
\label{sec:overview-part-1}

Consider the following scenario:
A hardware designer
  is designing a large hardware block
  for the Xilinx UltraScale+ family of FPGAs.
The designer is specifically aiming to use 
  the UltraScale+'s specialized DSP48E2 units,
  which 
  can implement 
  combined multiplication, arithmetic,
  and logic operations, as 
  captured in this 
  simplified
  block diagram~\cite{userguide}:
\begin{center}
\includegraphics[width=.75\columnwidth]{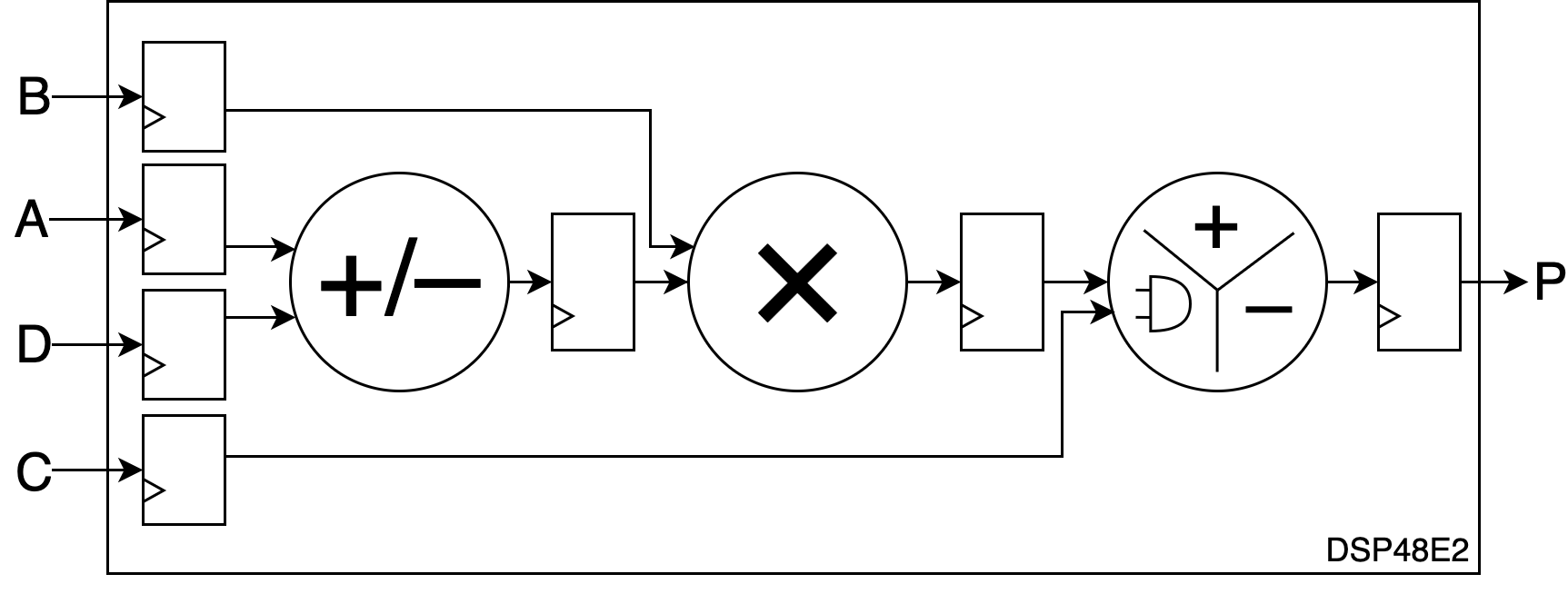}
\end{center}
The designer's hardware block
  involves the computation
  \texttt{(d+a)*b\&c},
  which the manual states is implementable with a single DSP.
In particular, suppose the design
  consists of four separate instances of the following computation:
\begin{minted}[fontsize=\footnotesize]{verilog}
 for(i=0; i<4; i++) begin
  r[i] <= (d[i] + a[i]) * b[i] & c[i];
 end
\end{minted}
It would be reasonable for the designer
  to expect the design
  to use a total of four DSPs.

\textit{\textbf{Current tools fail.}}
After compiling the design
  with existing tools,
  the designer is frustrated
  to find that the compiler returns a design
  that uses more
  resources than anticipated.
It does use four DSPs,
  but it also uses 128
  \textit{registers}
  (which hold state)
  and 64 \textit{lookup tables} 
  (LUTs, which implement logic functions).
\textbf{The compiler has thus failed to 
  fully utilize
  the DSP}---%
  it has not configured a DSP48E2
  to implement
  \texttt{(d+a)*b\&c} but has instead
  implemented a portion of the
  computation
  with LUTs and registers.
The designer now faces a choice:
  either accept the result or attempt to coax the compiler
  into returning a more optimal design.

\textit{\textbf{Coaxing the compiler, to no avail.}}
Though many 
may choose to accept a less optimal result,
 this tenacious\footnote{
This may not be purely a personal preference. 
For example, a hardware design simply may not fit on an FPGA
  without manual optimizations!}
  tries to coax
  the compiler into giving the 
  expected results
  by 
  placing the computation
  of interest into a separate module:
  
\begin{minted}[fontsize=\footnotesize]{verilog}
// add_mul_and.v: computes (a+b)*c&d in two clock cycles.
module add_mul_and(input clk, input [15:0] a, b, c, d,
                   output reg [15:0] out);
 reg [15:0] r;
 always @ (posedge clk) begin
   r <= (a+b)*c&d; out <= r;
 end
endmodule
\end{minted}

\noindent
This allows the designer
  to  apply
  specific optimizations while mapping
  the module---%
  a process we call \textit{partial design mapping}.
They attempt various strategies,
  including
  annotating the module with
  Xilinx's \texttt{use\_dsp} Verilog attribute
  (to force the compiler to use a DSP where possible)
  and using a different
  synthesis directive
  (to apply a more 
  resource-intensive synthesis
  procedure).
\textbf{Despite these efforts,
  the compiler still cannot 
  map the design
  to a single DSP,}
  instead using one DSP, 
  32 registers, and 16 LUTs.
Again, the designer must decide:
  give up and accept
  suboptimal results,
  or press on?

\textit{\textbf{Manual compilation.}}
The hardware designer presses on and
  now has only one option remaining:
  manually instantiating a DSP48E2
  with the desired behavior.
Skimming through the daunting 75-page DSP48E2's online user manual,
  the designer quickly discovers that
  configuring even
  the ``pre-add'' \texttt{a+b}
  requires correctly setting 
  multiple ports and parameters
  (\texttt{INMODE}, \texttt{AMULTSEL}, \texttt{BMULTSEL}, and \texttt{PREADDINSEL}),
  whose descriptions span 10+ pages and multiple tables.
Correctly configuring the subsequent multiplier
  and logic unit proves even more difficult
  and time-consuming.
After configuring the computational units,
  the designer must still manually ensure
  correct pipelining
  of the 10+ pipeline registers.
After hours
  of frustration,
  a configuration is found that 
  seems to work, which the designer
  inserts into the design.
Precious time has been wasted,
  most of which will need to be repeated
  to configure the DSP again.
Making matters worse,
  \textbf{the designer has no formal guarantees
  about the correctness of this DSP configuration.}
It may work in a few simulated test cases,
  but are there corner cases
  that have been missed?\tighten



\subsection{Compiling a Design to a DSP with \lr}
\label{sec:overview-part-2}
\begin{figure}
    \centering
    \hspace{-1cm}%
    \includegraphics[width=0.85\columnwidth]{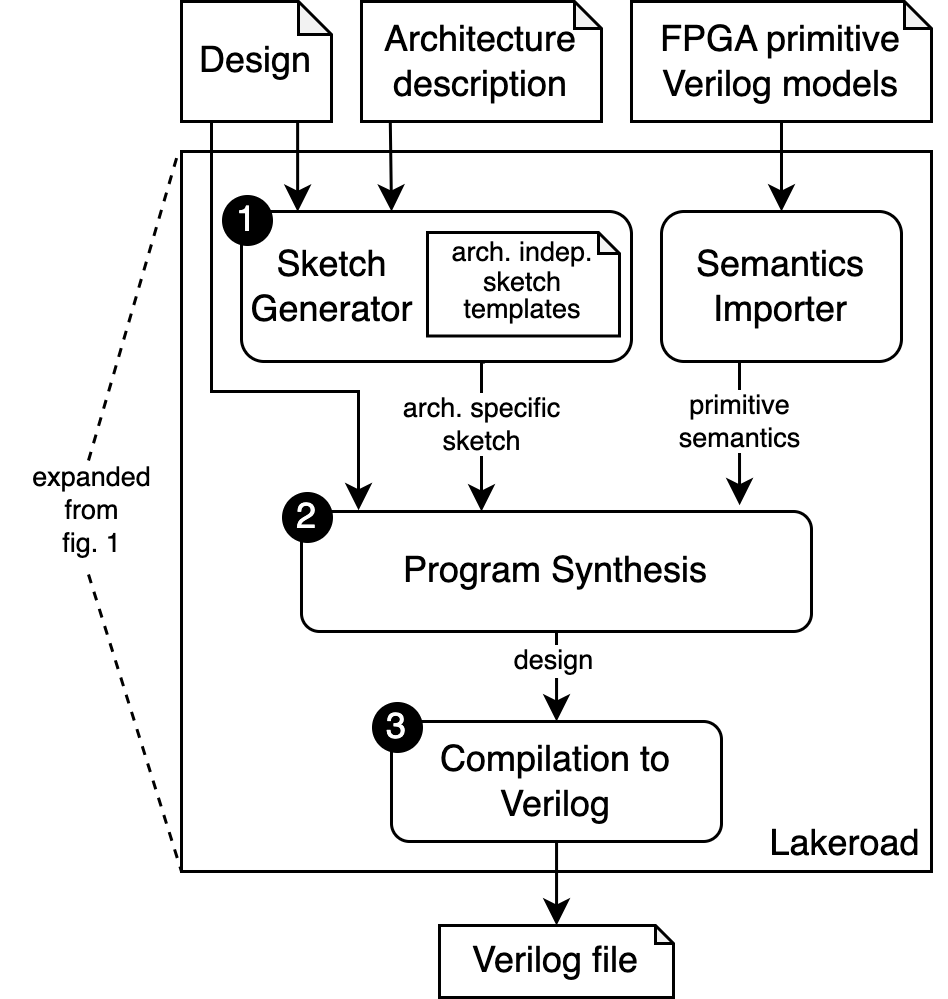}
   \caption{The components within \lr.}
    \label{fig:lakeroad-diagram}
\end{figure}

\lr can save hardware designers
  the great effort involved
  in manual DSP configuration
  while also providing correctness guarantees.
Let us imagine how the designer 
  in this example,
  frustrated by conventional tools,
  can instead proceed using \lr
  during partial design mapping.
After putting 
  \texttt{add\_mul\_and}
  into its own module,
  the designer calls
  \lr from the command line:
\begin{minted}[fontsize=\footnotesize]{bash}
$ lakeroad --template dsp \
           --arch-desc xilinx-ultrascale-plus.yml \
           add_mul_and.v
\end{minted}
The \texttt{lakeroad} command outputs
  \texttt{add\_mul\_and\_impl},
  an implementation
  of \texttt{add\_mul\_and}
  that uses a single UltraScale+ DSP48E2:
\begin{minted}[fontsize=\scriptsize]{verilog}
module add_mul_and_impl(input clk, input [15:0] a, b, c, d,
                        output [15:0] out);
 DSP48E2 #(
  .ACASCREG(32'd0), .ADREG(32'd0), .ALUMODEREG(32'd0),
  .AMULTSEL("AD"), .AREG(32'd0), .AUTORESET_PATDET("NO_RESET"), 
  // ...plus 30+ more parameters
 ) DSP48E2_0 (
   .A({ 14'h0000, a }), .ACIN(30'h00000000), .ALUMODE(4'hc),
   .B({ 2'h0, b }), .BCIN(18'h00000), .C({ 32'h00000000, c }),
   .CARRYCASCIN(1'h0), .CARRYIN(1'h0), .CARRYINSEL(3'h6),
   // ...plus 30+ more ports
 );
endmodule
\end{minted}
Unlike current compilers,
  \lr has produced an implementation
  using a single DSP48E2
  by utilizing more of the DSP's features.
Importantly, this compiled design
  is also formally guaranteed
  to implement
  the input \texttt{add\_mul\_and}
  design.

How does \lr provide 
  \textit{verified, more complete} support
  for the DSP48E2 over existing tools?  
At the core of \lr's
  correctness and completeness
  is
  \textit{sketch-guided program synthesis},
  a technique that  
  begins with a program \textit{sketch}, which
  captures a rough
  outline of a program
  and uses
  automated reasoning tools
  (e.g., SMT solvers)
  to fill in the sketch's \textit{holes.}
As shown in \cref{fig:lakeroad-diagram}, 
  \lr uses the following three-step process  
  to generate an efficient and correct
  DSP48E2 implementation
  of
  the \texttt{add\_mul\_and}
  design.

\textit{\textbf{Step 1: Generating a Sketch.}}
In
  the \texttt{add\_mul\_and}
  example,
  \lr
  generates
  the following sketch,
  which we refer to as \texttt{sketch}:\footnote{
Though this example is presented in a Verilog-like language,
  \lr's sketches are actually encoded in a Racket DSL that resembles structural Verilog.}
\begin{minted}[fontsize=\scriptsize]{verilog}
module sketch(input clk, input [15:0] a, b, c, d,
              output [15:0] out);
 DSP48E2 #(
  .ACASCREG(??), .ADREG(??), .ALUMODEREG(??), .AMULTSEL(??), 
  .AREG(??), .AUTORESET_PATDET(??), ...
 ) DSP48E2_0 (
   .A({ 14'h0000, a }), .ACIN(??), .ALUMODE(??), 
   .B({ 2'h0, b }), .BCIN(??), .C({ 32'h00000000, c }),
   .CARRYCASCIN(??),  .CARRYIN(??), .CARRYINSEL(??), ...
 );
endmodule
\end{minted}
This sketch consists of a single DSP48E2 instance
  with \textit{holes} 
  (represented by \texttt{??}) 
  serving as placeholders for most of its ports and parameters.
It is easy to see  
  the parallels
  between \texttt{sketch}
  and
  \texttt{add\_mul\_and\_impl};
  \texttt{sketch}
  is simply 
  \texttt{add\_mul\_and\_impl}
  with holes.
But how does \lr generate \texttt{sketch}
  in the first place?

To maximize portability across architectures,
  \lr does not store sketches 
  like \texttt{sketch}
  directly; 
  instead, it \textit{generates} sketches
  from architecture-independent
  \textbf{sketch templates.}
Instead of storing
  the preceding UltraScale+--specific sketch,
  \lr generates the sketch
  from the DSP sketch template, which
  the designer has chosen to use 
  with the \mbox{\texttt{-{}-template dsp}} flag.
A simplified form of this template looks like the following:
\begin{minted}[fontsize=\scriptsize]{verilog}
module dsp_sketch_template(input clk,
                           input [n-1:0] a, b, c, d,
                           output [n-1:0] out);
 DSP dsp_instance(.clk(clk), .A(a), .B(b), .C(c), .D(d), .out(out));
endmodule
\end{minted}
Sketch templates
  capture hardware design patterns
that are common across FPGA architectures
  in an
  architecture-independent way.
\texttt{dsp\_sketch\_template},
  for example, 
  captures
  a basic pattern, i.e., 
  instantiating a single DSP.
\lr includes 
  sketch templates of varying complexity,
  from the simplicity of the one above 
  to the complexity of LUT-based multipliers.
Though new sketch templates
  can be added easily,
  in most cases
  (as in this example)
  users can simply apply
  \lr's provided templates.

To specialize \texttt{dsp\_sketch\_template}
  into \texttt{sketch},
  \lr translates
  the sketch template's generic \texttt{DSP}
  \textbf{primitive interface}
  into an UltraScale+--specific
  DSP48E2
  using the UltraScale+
  \textbf{architecture description.}
The generic \texttt{DSP}
  module is an instance of a
  \textbf{primitive interface:} 
  a \lr-introduced abstraction that 
  captures the similarities
  between primitives across 
  diverse FPGA architectures.
For example, 
  \lr's DSP primitive interface
  captures the facts that
  DSPs on all FPGA platforms
  generally have two to four data inputs
  (captured by \texttt{A}--\texttt{D})
  and a clock input
  (captured by \texttt{clk}).
To convert the
  sketch template's
  DSP primitive interface instance
  into a DSP48E2,
  \lr utilizes the
  Xilinx UltraScale+ architecture description,
  which the designer has pointed to with the 
  \texttt{-{}-arch-desc xilinx-ultrascale-plus.yml}
  flag.
An \textbf{architecture description}
  specifies how \lr's various
  primitive interfaces
  are implemented for a given architecture.
The following simplified snippet of the UltraScale+
  architecture description, for example,
  tells \lr that,
  when generating a sketch for 
  UltraScale+,
  instances of the DSP
  primitive interface
  should be implemented with a DSP48E2:
\begin{minted}[fontsize=\scriptsize]{yaml}
- interface: {name: DSP, params: { out-width: 48, a-width: 30, ...}} 
  holes: [?ACASCREG, ?ADREG, ?ALUMODEREG, ?AREG, ...]
  implementation:
    module: DSP48E2
    ports: [{ name: A, bitwidth: 30, value: A }, ...]
    parameters: [{ name: ACASCREG, value: ?ACASCREG }, ...]
    outputs: { O : P }
\end{minted}
Thus, while converting
  \texttt{dsp\_sketch\_template}
  into
  \texttt{sketch},
  \lr reads this architecture description
  and converts the single DSP instance
  into a DSP48E2,
  filling the ports and parameters with the concrete values
  and holes
  contained in the architecture description.
Architecture descriptions
  are usually short (100-400 LoC)
  and
  written only once per FPGA architecture; 
  \lr already contains such descriptions
  for Xilinx UltraScale+,
  Lattice ECP5,
  Intel Cyclone 10 LP,
  and the open-source FPGA SOFA~\cite{sofa}.

To generate a sketch,
  \lr takes an architecture-independent
  sketch template
  and specializes it using an
  architecture description.
Once the sketch is ready,
  the designer can move on to synthesis.

\textit{\textbf{Step 2: Program Synthesis.}}
  The next step 
  fills
  in the holes to generate
  a complete, correct
  hardware design,
  which is done automatically
  using a technique called
  \textbf{program synthesis}.
\textit{Program synthesis} is the process of
  using automated reasoning tools
  (like SMT solvers) 
  to generate correct programs
  by encoding program generation
  as a constraint solving problem.
In our \texttt{add\_mul\_and} example,
  \lr, aided  by
  Rosette~\cite{torlak2013growing,torlak2014lightweight}, 
  generates a query like the following:\footnote{
  We formalize this synthesis query and explain it precisely in \cref{sec:formalization}.
}
\footnotesize
\begin{multline*}
\exists \ \mathtt{ACASCREG}, \mathtt{ADREG}, ...\ . \forall \mathrm{inputs}. \\
  \texttt{add\_mul\_and}(\mathrm{inputs}) 
  = \texttt{sketch}(\mathrm{inputs}, \mathtt{ACASCREG}, \mathtt{ADREG}, ...)
\end{multline*}
\normalsize
The query asks:
  are there concrete values for
  \texttt{ACASCREG}, \texttt{ADREG}, etc.,
  that will make our sketch's behavior
  equivalent to the input design's behavior
  on all inputs?
If the solver finds such values,
\lr can use them to fill the holes
  in the sketch
  and produce a compiled design.
However, if \lr tries to pass the preceding formula
  to an SMT solver,
  the solver will throw an error since
  the query is not expressed
  at a level 
 it understands, viz., %
  as equalities
  between bitvector expressions,
  using simple Boolean or arithmetic
  operations.
While it is conceivable that \texttt{add\_mul\_and}
  could be converted to
  a bitvector expression
  since its core computation is already
  expressed as
  \texttt{(a+b)*c\&d},
it is unclear how to express
  \texttt{sketch}
  as an expression over bitvectors.
In particular, \lr must express
  bitvector-level semantics
  for Xilinx's DSP48E2 primitive.

To generate bitvector-level semantics
  for complex FPGA primitives,
  \lr introduces the concept of
  \textbf{semantics extraction}.
Rather than requiring
  manual effort to encode the semantics
  of the underlying hardware,
  which is notoriously difficult
  even for experts~\cite{Bernstein2021WhatAT},
  \lr's key insight is that these challenges
  can be avoided altogether
  by extracting low-level semantics
  directly from
  vendor-supplied simulation and verification models.
\lr builds on internal passes in 
  Yosys~\cite{wolf2013yosys} 
  to automatically extract
  solver-ready semantics from these vendor-provided HDL models,
  which we detail in \cref{sec:implementation:importing-semantics}.
For the
  \texttt{add\_mul\_and} example,
  the DSP48E2's semantics have already been imported into \lr.
Semantics need to be imported only when
  adding support for a new architecture, i.e., %
  about as infrequently
  as writing a new architecture description.
In most cases,
  \lr users can rely on already-imported semantics.
  
With the sketch generated
  and the DSP48E2's semantics imported,
  program synthesis can begin.
\lr utilizes Rosette to drive program synthesis,
  as detailed in \cref{sec:formalization}.
In our example, Rosette returns
  a configuration for the DSP48E2.
The last step, then, is to convert
  the compiled design
  to Verilog.
  
\textit{\textbf{Step 3: Compilation to Verilog.}}
Compiling \lr's internal representation into Verilog
  is a purely 
  one-to-one syntactic mapping;
  no optimizations are done at this stage,
  reducing the likelihood that bugs could be inserted.
In our example,
  the final Verilog produced
  results in the
  \texttt{add\_mul\_and\_impl}
  we saw at the start of \cref{sec:overview-part-2}.

\textit{\textbf{In summary.}}
\lr
  delivered 
  an implementation
  of the designer's
  \texttt{add\_mul\_and}
  module,
  improving upon both state-of-the-art compilers
  and manual approaches
  in multiple ways.
\lr's implementation
  is significantly more resource-efficient
  than the state-of-the-art compiler's---%
  one DSP versus one DSP,
  32 registers, and 16 LUTs.
\lr delivered its implementation
  in mere seconds,
  compared to the hours to days
  of work
  that manually instantiating
  a DSP might take.
Lastly,
  \lr's implementation
  is formally guaranteed to be correct.
Meanwhile, \lr did all of this
  while requiring no input from the user
  other than the Verilog to be compiled.



\section{Formalization}
\label{sec:formalization}

\newsavebox\boxlet
\newsavebox\boxassign
\newsavebox\boxin
\newsavebox\boxreg
\newsavebox\boxprim
\savebox{\boxlet}{\lstinline[language=thelang]!let!}
\savebox{\boxassign}{\lstinline[language=thelang]!:=!}
\savebox{\boxin}{\lstinline[language=thelang]!in!}
\savebox{\boxreg}{\lstinline[language=thelang]!Reg!}
\savebox{\boxprim}{\lstinline[language=thelang]!Prim!}

\newcommand{\Prim}[0]{\lstinline[language=thelang]{Prim}\xspace}

\newcommand{\Reg}[0]{\lstinline[language=thelang]{Reg}\xspace}

\newcommand{\Let}[0]{\lstinline[language=thelang]{let}\xspace}

\begin{figure*}

\begin{minipage}{.35\textwidth}
\begin{tabular}{l@{\hspace{.5em}}l}
$\SynProg$     & $\defn$ $\langle$ $\SynId$, $\seq{\SynId, \SynNode}*\rangle$\\[0.5em]
$\SynNode$       & $\defn \SynBv\ b$  | $\SynVar\ x$ \\
                 & \ \   | $\Op\ op\ \SynId$* \\
                 & \ \   | \IRReg{} $\SynId$ $(\SynBv~b)$ \\
                 & \ \   | \IRPrim{} $\mathsf{binds}$ $\SynProg$ \\
                 & \ \   | $\blacksquare_x$ \\[0.5em]
    
%
\end{tabular}
\end{minipage}%
\begin{minipage}{.25\textwidth}
\footnotesize{}
\begin{tabular}{l l}
$\SynId$  & $id  \in \mathbb{N}$ \\  
Bitvectors  & $b  \in \mathbb{BV}$ \\
Variables & $x  \in \text{LegalVarNames}$\\[0.5em]
Operators     & $op \in$ $\OpBv  \cup  \OpWire$\\[0.5em]
$\mathsf{binds}$ & $bs \in (\text{Variables} \rightharpoonup \SynId)$
\end{tabular}
\end{minipage}
\begin{minipage}{.4\textwidth}
\footnotesize{}
\begin{tabular}{l l}
\;\;\; Wire op & $\OpWire =$ \lstinline[language=thelang,mathescape]!$\{$concat$,$extract$, \ldots \}$! \\
\;\;\; Non-wire op & $\OpBv =$ \lstinline[language=thelang,mathescape]!$\{+, -, \times, \ldots \}$! \\
\end{tabular}
\end{minipage}

\caption{Syntax of $\UberLang$. $\blacksquare_x$ is a syntactic hole, labeled with variable $x$. $A \rightharpoonup B$ denotes the set of partial functions from $A$ to $B$.}
\label{fig:syntax}
\end{figure*}



We now formalize \lr{} 
  with functions $\lrfn$ and
  $\lrfnbmc$,
  and use these models
  to argue for the correctness
  and partial completeness of \lr{}. 
We first define  
  $\lrfn$ (\cref{subsec:the-lr-function}) and then 
motivate and define
  the language $\UberLang{}$,  
  specify its syntax and semantics,
  and define behavioral ($\SpecLang$), structural ($\ImplLang$), and sketch ($\SketchLang$) sublanguages (\cref{subsec:lrir-syntax-and-semantics}).
{
We next explain the
  underlying queries
  \lr{} uses to
  synthesize hardware programs
  that meet the desired specification
} (\cref{sec:formalization-program-synthesis}).
  We demonstrate the correctness
  and partial completeness of $\lrfn$,
  enumerate our Trusted Computing Base
  (\cref{subsec:lr-correctness-and-completeness}) and 
extend $\lrfn$ to $\lrfnbmc$,
    which ensures the generated program's
    semantics matches the design over multiple
    timesteps (\cref{subsec:lrfn-bmc}).
Finally,
 we highlight potential future
 applications that could be
 built on this section's formalization
 (\cref{sec:sem-beyond}).

\subsection{The \lr Function $\lrfn$}
\label{subsec:the-lr-function}

We model the execution of \lr
  with the partial function
  \small
\[\lrfn: \Sketch \times \SpecLang \times \Time \rightharpoonup \ImplLang,\]
\normalsize
where $\lrfn(\Psi, d, t)$
  invokes Rosette
  to synthesize a $t$-cycle
  implementation
  of behavioral design $d$ 
  using sketch $\Psi$ to guide
  the search,
  where a $t$-cycle
  implementation of $d$ is a program
  that is equivalent to $d$ at clock cycle
  $t$.
By not requiring program equivalence before
    clock cycle $t$ we allow the
    synthesized program's semantics to
    differ from the design during 
    an initialization period
    (e.g., as the pipeline is being filled).
To get guarantees beyond a single
    point in time $t$, we generalize
    $\lrfn$ to $\lrfnbmc$,
    which synthesizes a program
    that is equivalent to the design
    from time $t$ to $t + n$.
We formalize a sketch $\Psi \in \Sketch$  as a tuple
  $(\psi, h)$,
  where $\psi$ is a program in $\SketchLang$ 
  and $h$ is a map
  from the holes in $\psi$
  to a finite set of valid hole-free
  nodes in $\ImplLang$
  that can be used to fill
  the mapped hole.
This mapping $h$ is handled implicitly by Rosette's
  \texttt{choose} and \texttt{hole} constructs and
  need not be explicitly specified by the
  sketch writer.
We write $\lrfn(\Psi, d, t) = p$ to 
    indicate that synthesis succeeded and produced
    \lr{} program $p$.
However, it is possible that sketch $\Psi$
    cannot implement $d$, in which case
    the synthesis fails
    (i.e., returns UNSAT)
    and
    $\lrfn$ does not return anything.
Design $d$ belongs to 
  $\UberLang$'s behavioral fragment,
  $\SpecLang$ (see
  \cref{subsec:lrir-syntax-and-semantics}).
When $t = 0$, $\lrfn$ 
  synthesizes a 
  \textit{combinational design}; when 
$t > 0$, $\lrfn$ 
  synthesizes a \textit{sequential design}
  over $t$ clock cycles.
The rest of this section 
  considers sequential design synthesis 
  since its combinational counterpart is 
  a special case covered
  by our general approach.

\subsection{Defining $\UberLang$}
\label{subsec:lrir-syntax-and-semantics}

\lr uses the 
  $\UberLang$ language to
  translate behavioral
  HDL programs
  to structural, 
  hardware-specific
  HDL programs.
To facilitate this
  translation, we
  designed $\UberLang$
  to satisfy the
  following properties:
\begin{enumerate}[label=\textbf{P\arabic*}.]
    \item \textit{Easy translation to/from HDLs:}
        we must be able
        to translate
        designs from 
        a behavioral HDL
        to $\UberLang$
        and translate
        synthesized implementations
        to a structural HDL.
        
    \item \textit{Support parallel stateful execution:}
        FPGA designs
        consist of
        potentially stateful elements
        executing in parallel.
        $\UberLang$ must
        allow unambiguous
        parallel execution.

    \item \textit{Support graph-like program structures:}
        An FPGA component's outputs
        can be wired to
        multiple other components,
        including back
        to itself.
        This means that 
        FPGA programs can
        form arbitrary
        graphs, and $\UberLang$ must
        be able to express this.
        
    \item \textit{Support for sequential designs:}
        $\UberLang$ must handle designs
        that run over multiple clock cycles.
        
    \item \textit{Support for different architectures:}
        $\UberLang$ must handle FPGA components
        from different architectures.
\end{enumerate}
We describe how $\UberLang$ satisfies
  P1-P5 when we 
  define its syntax and semantics
  in the following subsections.

\subsubsection{$\UberLang$'s Syntax}
\label{subsubsec:syntax}

\Cref{fig:syntax} shows the $\UberLang$ syntax.
An $\UberLang$ program $\SynProg$
  consists of a root node ID
  and a graph of nodes,
  each of which is
  referred to by its ID.
A \textit{node} 
  can be 
  a constant bitvector,
  input variable,
  combinational (pure) operator,
  sequential (stateful) register,
  primitive,
  or hole.
Given a program 
  $p = (r, \langle id_1, {node}_1\rangle
           \ldots
           \langle id_n, {node}_n\rangle)$, 
  we use the notation 
  $p.root = r$, 
  $p.ids = \left\{id_1, \ldots, id_n\right\}$,
  and
  $p[id_i] = {node}_i$.
We define the free variables of a program $p.fv = \{x_i\}$ as the set of variable names occurring in $p$'s nodes of the form $(\SynVar\ x_i)$.\footnote{Note that this does not include variables of sub-programs occurring recursively inside of \IRPrim nodes.}  Finally, we use the notation $p.all\_ids$ for $p.ids$ together with $p'.all\_ids$ of any subprogram $p'$ of $p$ ($p'$ is a subprogram of $p$ if $\exists j, node_j = \textrm{\IRPrim}~bs~p'$).

Given a node $n$,
  we specify its inputs
  with the following function:
  \small
\begin{align*}
  &\textsc{inputs}(\SynBv\ b) = \{\}, \\
  &\textsc{inputs}(\SynVar\ x) = \{\}, \\
  &\textsc{inputs}(\Op\ op\ i_1 \ldots i_n)
     = \{i_1,\, \ldots,\, i_n\} \\
  &\textsc{inputs}(\text{\IRReg}\ i\ b_{init})
     = \{i\} \\
  &\textsc{inputs}(\text{\IRPrim}\ 
                     bs\ p')
                = \{bs[x]\ |\ x \in \textrm{domain}(bs)\}
\end{align*}
\normalsize
Note that we use 
    $A \rightharpoonup B$ 
    to denote the set
    of partial functions
    from $A$ to $B$;
    given 
    $bs \in A \rightharpoonup B$, 
    we write $\textrm{domain}(bs)$
    to denote
    the set of 
    $x\in A$ s.t.
    $bs[x]$ is defined.


A program $p$ is well-formed
  if and only if
  all the following
  conditions hold:

\begin{enumerate}[label=\bfseries{W\arabic*}.]
  \item $p.root \in p.ids$;

  \item All ids are unique and distinct. (i.e. for any sub-program $p'$, $p.ids$ and $p'.all\_ids$ are disjoint, and for any two sub-programs $p'$ and $p''$, $p'.all\_ids$ is disjoint from $p''.all\_ids$.)
  
  \item The inputs of all nodes in $p$ are ids of other nodes in $p$:
     $\forall id \in p.ids$, 
     $\text{inputs}(p[id]) \subseteq p.ids$;

  \item All primitive nodes contain 
      well-formed programs;
  \item All primitive nodes bind exactly their free variables; i.e., for $\text{\IRPrim}\ bs\ p'$, $\textrm{domain}(bs) = p'.fv$; and
        
  \item Program $p$ is free of combinational loops (formalized below in \cref{property:free-of-combinational-loops}).
        %
\end{enumerate}

\begin{property}[Free of Combinational Loops]
\label{property:free-of-combinational-loops}
Formally, a program $p$ is free of combinational
loops if there exists a function
$w : p.all\_ids \to \mathbb{N}$, that satisfies the following properties (collectively ``monotonicity''):
\begin{enumerate}
\item If $p[id] = \text{\IRReg{}}~\_~\_$, then $w(id) = 0$;
\item If $p[id] = \text{\IRPrim{}}~bs~p'$, then $w(id) > w(p'.root)$;
\item if $p[id] = \text{\IRPrim{}}~bs~p'$ and $p'[id'] = Var~x$, \\ then $w(id') > w(bs[x])$; and
\item Otherwise (e.g., $p[id] = \Op~op~ids^{*}$), \\
    if $id' \in \textsc{inputs}(p[id])$,
    then $w(id) > w(id')$. 
\end{enumerate}
\end{property}
\noindent The function $w$ acts as a witness to the 
absence of combinational loops because it is
impossible to define a strictly monotonic function without acyclicity.
We consider only well-formed
  $\UberLang$ programs.

%

$\SynBv$, $\SynVar$, and $\Op$ nodes
  encode bitvectors, variables, and
  operators.

\Reg $i_{data}\ b_{init}$ nodes
  let \UberLang implement
  sequential designs (P4).
$i_{data}$ is the 
  register's data input,
  which updates the stored
  value at the positive edge
  of each clock cycle,
  and $b_{init}$ is the
  register's initialization value.

\IRPrim{} $bs\ p$ nodes
  let \UberLang programs
  use hardware-specific components
  from different architectures (P5).
The $bs$ component is a \textit{variable map},
  mapping $\SynVar$s to input $\SynId$s.
The $p$ component is an $\UberLang$ program
  that defines the semantics of the
  hardware primitive.
A \Prim node also carries some metadata 
  used during compilation
  to a structural
  HDL, which we omit 
  for clarity.\tighten

$\SpecLang$ is the concrete
  \textit{behavioral}
  fragment of $\UberLang$ used for
  writing specifications; it 
  is formed by
  excluding \Prim
  nodes and holes
  from $\UberLang$.
  
$\ImplLang$ is the concrete
  \textit{structural}
  fragment of $\UberLang$ used
  for lowering $\UberLang$
  to structural HDLs; it 
  is formed
  by excluding \Reg nodes, $\Op$ nodes,
  and holes from $\UberLang$,
  with the following
  exception:
  the $p$ term
  in $\usebox{\boxprim}\ bs\ p$ must
  always be from the $\SpecLang$ since it 
  is used to specify the semantics
  of the \Prim node
  to the synthesis engine.
The behavioral node $p$
  is not used during
  compilation to HDL,
  and this behavioral
  expression does not 
  propagate to the
  structural HDL output.

$\SketchLang$ is another sublanguage of $\UberLang$ that is $\ImplLang$ but also including holes. 
Let $s$ be a program in $\SketchLang$
  with holes 
  $\blacksquare_{x_1}, \ldots, \blacksquare_{x_k}$.
These holes can be \emph{filled}
  with nodes $n_1, \ldots, n_k$
  in $\ImplLang$ by replacing
  each hole $\blacksquare_{x_i}$ 
  with its 
  corresponding node $n_i$
  to obtain a complete
  $\ImplLang$ program,
  denoted by $s[\blacksquare_{x_1} \mapsto n_1, \ldots]$.

The simplicity of this syntax
  makes translating to and from
  HDLs straightforward (P1).
\cref{sec:implementation}
  describes how \lr{} implements the
  translations to and from HDLs.

\subsubsection{$\UberLang$'s Semantics}
\label{subsubsec:semantics}

\begin{figure}


\small


\begin{flalign*}
  & \Time \hspace{0.2cm} t\in\mathbb{N} \hspace{1cm}  \textsf{Env} \hspace{0.2cm} e \in (\SynVar \rightharpoonup \Time \to \SynBv) \\
  &\vspace{0.5cm} \\
  &\textsc{Interp}\ :\ \SynProg \to \textsf{Env} \to \Time \to \SynNode \to \SynBv\\
  &\textsc{Interp}\ p\ e\ t\ (\SynBv\ b)\ = b&& \\
  &\textsc{Interp}\ p\ e\ t\ (\SynVar\ x)\ = e\ x\ t&& \\
  &\textsc{Interp}\ p\ e\ 0\ (\text{\IRReg}\ \_\ init)\ = init&& \\
  &\textsc{Interp}\ p\ e\ (t + 1)\ (\text{\IRReg}\ id\ \_)\ = \textsc{Interp}\ p\ e\ t\ p[id]&& \\
  &\textsc{Interp}\ p\ e\ t\ (\Op\ \mathsf{op}\ ids)\ = \llbracket op\rrbracket \ (\text{map}\ (\lambda id \ .\ \textsc{Interp}\ p\ e\ t\ p[id])\ ids)&& \\
  &\textsc{Interp}\ p\ e\ t\ (\text{\IRPrim}\ bs\ p')\ =&& \\
  &\quad \text{let}\ e' = \lambda x, t'\ .\ \textsc{Interp}\ p\ e\ t' \left(p[bs\ x] \right) \text{in}\\
  &\quad \textsc{Interp}\ p'\ e'\ t\ p'[p'.root]
\end{flalign*}



\caption{\lr's semantics as pseudocode.}
\label{fig:lr-interpreter-pseudocode}
\end{figure}

Before discussing the formal
  semantics of \UberLang,
  we present key 
  definitions.
We assume 
  a \textit{bitvector type}
  and, for simplicity,
  we elide bitvector
  widths.
We represent \textit{time} as a
  natural number.
A \textit{stream} is a function from $\Time$
  to bitvectors.
An \textit{environment} is
  a map from
  variable names
  to streams.

We give the semantics 
  for $\UberLang$ as an interpreter in 
  \cref{fig:lr-interpreter-pseudocode}.
We define the function \textsc{Interp} to interpret a program $p$ in environment $e$ at time $t$ and node $n$.
We do not define semantics for holes,
  as they are intended to be replaced
  by other constructs with well-defined semantics.

Most of the rules are straightforward.
A bitvector $\SynBv\ b$
  evaluates to
  its backing bitvector value $b$.
A variable node $\SynVar\ x$ 
  in an environment $e$
  at time $t$
  evaluates to
  the value returned 
  by the stream associated
  with $x$ in $e$
  at time $t$; 
using function notation, this is
  denoted by $e\ x\ t$.
A $k$-ary operator node
  $\Op{}\ op\ i_1\ldots i_k$ recursively
  interprets each operand in the current
  environment at the current time 
  and then applies $op$'s semantics,
  denoted $\llbracket op \rrbracket$,
  to the resulting values.
A register \IRReg $id\ b_{init}$ has two cases 
  depending on the current time: 
at time $t = 0$, a register evaluates
  to its initial bitvector value $b_{init}$; 
at nonzero times $t + 1$, a register evaluates
  to the value produced by the input $i$
  at the \textit{previous} timestep $t$.
A primitive \IRPrim{} $bs\ p'$
  in environment $e$ at time $t$
  is evaluated by interpreting
  the program $p'$ under
  the fresh environment $e'$
  formed by the binding map $bs$.

\subsection{Program Synthesis}
\label{sec:formalization-program-synthesis}


$\lrfn$ performs sketch-based program
  synthesis~\cite{solar2008program}.
Operationally, we implement
  the \textsc{Interp} function from 
  \Cref{fig:lr-interpreter-pseudocode}
  in Rosette, a solver-aided host
  language~\cite{torlak2014lightweight}.
Let sketch $\Psi = (\psi, h) \in \Sketch$, where
  $\psi \in \SketchLang$ has holes
  $\blacksquare_{x_{i}}$
  and $h$ maps $\psi$'s holes
  to the set of structural nodes
  that can legally fill the mapped hole.
Given a design $d$,
  we query Rosette if there
  are nodes 
  $n_1, n_2, \ldots n_k$ 
  such that 
  $n_i \in h[\blacksquare_{x_i}]$
  and
  $p = \Psi[\blacksquare_{x_1} \mapsto n_1, \ldots]$
   is well-formed and equivalent to $d$
   (i.e., we ask Rosette to fill
   each hole with a node associated with the node in $h$).
Program equivalence between well-formed
  programs $p$ and $d$ at time
  $t$, written $p\cong_t d$, is defined as
  \begin{align*}
    & p.fv = d.fv\ \wedge \\
    &\forall e\ s.t.\ \textrm{domain}(e) = p.fv, \\
    &\quad\textsc{Interp}\ p\ e\ t\ p.root =
   \textsc{Interp}\ d\ e\ t\ d.root.
  \end{align*}
In \cref{subsec:lrfn-bmc}, we
  use bounded model checking to
  extend $\lrfn$'s guarantees
  beyond the single timestep
  at clock cycle $t$.

\subsection{Correctness and Completeness of $\lrfn$}
\label{subsec:lr-correctness-and-completeness}

Recall that the synthesis
  function $\lrfn$ is partial.
We say that $\lrfn$
  is \emph{correct} if
  it 
  returns a program
  $\lrfn(\Psi,d,t) = p$  where
  $p$ is
  a well-formed
  completion of $\Psi = (\psi, h)$,
  meaning $p = \Psi[\blacksquare_{x_1} \mapsto n_1, \ldots]$
            such that $n_i \in h[\blacksquare_i]$ for all $i$
  and $p\cong_t d$.

Furthermore, we say
  that $\lrfn$ is
  \emph{sketch-complete}
  if $\lrfn(\Psi,d,t)$
  is defined whenever
  there exists a
  well-formed completion
  $p$ of $\Psi$
  such that $p\cong_t d$.
That is, synthesis is
  correct if it never
  returns an erroneous result
  and sketch-complete
  if it returns a
  correct result
  whenever one exists.

{
We have implemented $\lrfn$
    with Rosette 
    (see \cref{sec:formalization-program-synthesis}),
    which guarantees our system is correct and complete
    under the following assumptions:
\begin{enumerate}
    \item Correctness of Rosette and underlying SMT solvers;
    \item That our encoding of \lr{} is bug-free;
    \item That the lowering of \textsc{Interp}
    to SMT formulas by  Rosette always terminates.
    This is possible when partial evaluation of \textsc{Interp} on arguments $p$, $t$ and $n$ terminates (independently of the value of $e$).
\end{enumerate}

}

\begin{lemma}
\label{lemma:interp-is-primitive-recursive}
Let $p$ be a well-formed program, 
    $e$ an environment,
    $t$ a \Time,
    and $n$ be a node belonging to $p$.
Then \textsc{Interp} is primitive recursive 
    (i.e. terminates) in the arguments $p$, $t$, and $n$.\tighten
\end{lemma}

\begin{proof}[Proof of Lemma~\ref{lemma:interp-is-primitive-recursive}]
Recall that a function $f(x,y,z)$ is
    primitive recursive in arguments $x$ and $y$
    (under a lexicographic ordering) 
    if in the definition of $f$ every
    recursive call $f(x',y',z')$ is made
    with values $(x',y')$ such that $x' < x$ 
    or $x' = x \wedge y' < y$.
    If $x$ and $y$
    are drawn from the natural numbers 
    (or another well-ordered set),
    then the recursion is guaranteed to terminate.

Under what order is 
  \textsc{Interp} primitive recursive?
Because our program
    is well-formed, it must be free
    of combinational loops (see \cref{property:free-of-combinational-loops}).
Formally, this means we have an acyclicity
    witness function $w : p.all\_ids \to \mathbb{N}$
    that monotonically increases in the direction of
    dataflow in our circuit.
Each node $n$ argument passed to \textsc{Interp}
    has an \textsf{Id} that is unique and distinct
    from the \textsf{Id}s used in $p$ or any of $p$'s
    subprograms (\textbf{W2});
    we denote this \textsf{Id} as $id_n$.
We can associate each $n$ argument
    to a recursive call of \textsc{Interp}
    with a number $w(id_n)$.
We claim that \textsc{Interp} is
    primitive recursive under
    the lexicographic ordering on $(t, w(id_n))$.
    
To prove this claim we need to demonstrate that
    if \textsc{Interp} with time and node
    arguments $t'$ and $n'$ makes a recursive
    call to  \textsc{Interp} with time and
    node arguments $t''$ and $n''$, then the following
    condition holds:
    \small
\begin{equation}
\label{eqn:primitive-recursion-condition}
    t'' < t' \vee \left(t'' = t' \wedge w(id_{n''}) < w(id_{n'})\right).
\end{equation}
\normalsize
To do this it suffices to examine each case of \textsc{Interp}'s definition.

When $n'$ is a $\SynBv$ constant,
    \textsc{Interp} makes no recursive calls,
    and the condition in \cref{eqn:primitive-recursion-condition}
    holds vacuously.
    
When $n'$ is a \IRReg{} node \textsc{Interp} either terminates
  (when $t'=0$) or makes a
  recursive call with time value $t'' = t' - 1$,
  maintaining the condition in \cref{eqn:primitive-recursion-condition}.
  
When $n'$ is an operator node, 
  \textsc{Interp} recursively interprets
  the operands with time arguments $t'' = t'$.
However, each operand's id $id''$ belongs to $\textsc{inputs}(n')$,
    and, by~\cref{property:free-of-combinational-loops},
    $w(id_{n'}) > w(id'')$,
    so our condition holds.
  
This leaves us with the less obvious
  cases in which $n'$ is either a \IRPrim or $\SynVar$,
  which work together in tandem.
When $n' = \text{\IRPrim{}}~bs~p'$,
    \textsc{Interp} makes a recursive
    call with node argument $p'.root$
    and time argument $t$.
By ~\cref{property:free-of-combinational-loops},
    $w(p'.root) < w(id_{n'})$,
    and the condition in \cref{eqn:primitive-recursion-condition} holds.
\textsc{Interp} also defines a new environment for execution
    of $p'$ via $\lambda$-abstraction, and this in turn
    will recursively invoke \textsc{Interp}.
These environments are only invoked by the rule for variables,
    which we handle presently.
    
When $n' = \SynVar~x$, the environment is 
    invoked on variable $x$.
Here, there are two possible cases.
First, we are interpreting the
    top-level program $p$. 
As this is the initial, top-level environment, there is no further recursion.
Second, we are
    interpreting a sub-program $p'$
    and $e'~x~t = \textsc{Interp}~p~e~t~(p[bs~x])$
    is actually a recursive call into the
    program $p$ one level up,
    with its environment $e$.
In this latter case,
    note that $w$ is defined such that
     $w(id_{p[bs~x]}) = w(bs~x) < w(id_{\SynVar~x})$
    (item 3 of \cref{property:free-of-combinational-loops}),
    satisfying our property.
All cases are complete.
\end{proof}


From this, we conclude that
  all possible substitutions for $\Psi$
  are attempted, and $\lrfn$ is sketch-complete.


\paragraph{Trusted Computing Base.}

The \textit{trusted computing base} (TCB) of a
  system is the set of components
  it assumes to be correct~\cite{MacKenzieComputingTrust}.
A bug anywhere in the TCB
  could cause the guarantees
  made by that system to be violated.
\lr's  TCB includes:
  Rosette and the underlying
      SAT/SMT solvers that Rosette queries
      (Bitwuzla, cvc5, Yices2, and STP);
  the internal Yosys passes \lr
      uses to extract primitive semantics
      and translate design specifications
      from behavioral Verilog into
      $\SpecLang$;
  the semantics for $\UberLang$,
    which we assume conservatively
    models non-cyclic (DAG) designs;
  our code to translate from
      the $\ImplLang$ to
      structural Verilog; and
  the vendor-provided Verilog
    simulation models for FPGA primitives.
Each TCB component 
  has also been thoroughly tested,
  as described in \cref{sec:evaluation}.
Importantly,
  sketches and sketch generation
  are \textit{not} in \lr's TCB: %
  even if there were a
  bug in \lr's sketch-related components,
  it would not violate
  \lr's correctness guarantees.

\subsection{Multiple Clock Cycle Guarantees with $\lrfnbmc$}
\label{subsec:lrfn-bmc}

The preceding completeness and 
    correctness properties for
    $\lrfn$ 
    guarantee that
    running the
    synthesized program
    $p$ and the design $d$
    for $t$ clock cycles
    produces the same output.
To extend this guarantee, \lr supports
    a form of
    bounded model checking, 
    where
    synthesis ensures that
    $p$ is semantically equivalent
    to $d$ for $c$ additional clock cycles
    starting at time $t$.
We formalize this
    with the function $\lrfnbmc$,
    which takes a sketch $\Psi$,
    a behavioral design $d$,
    a number of clock cycles $t$,
    and a model checking
    time bound $c \geq 0$
    and returns an implementation
    $p \in \ImplLang$
    that is equivalent to $d$
    at time steps $t, t+1, \ldots, t + c$.

Our correctness and completeness guarantees are
    similar to those for $\lrfn$:

    \small
  \begin{align*}
    & p.fv = d.fv\ \wedge \\
    &\forall e\ s.t.\ \textrm{domain}(e) = p.fv, \\
    &\quad\bigwedge_{i=t}^{i=t+c}
      {\textsc{Interp}\ p\ e\ i\ p.root = \textsc{Interp}\ d\ e\ i\ d.root}.
  \end{align*}
  \normalsize

\subsection{Beyond \lr}
\label{sec:sem-beyond}

$\UberLang$, its semantics,
  and the synthesis approach we describe here 
  are useful for applying program synthesis
  to other hardware design problems.
For example,
  the synthesis problem detailed above could be ``flipped''
  to decompile structural designs back 
  to higher-level behavorial designs,
  i.e., synthesizing from $\ImplLang$
  to an expression in $\SpecLang$.
Such decompilation has seen recent
  interest for
  recovering equivalent but faster-to-simulate
  models and for porting models across
  different architectures~\cite{sisco2023loop}.
As another example,
  the synthesis approach could be
  adapted to help port designs by
  synthesizing expressions in
  $\ImplLang$ that use one set of primitives
  on one architecture from 
  other designs in $\ImplLang$ that use
  a different set of primitives from
  a different architecture.
Thus, the formalization in this section
  transcends the particular challenges
  of FPGA technology and provides
  a reusable foundation for exploring
  a much broader range of hardware design challenges
  from a program synthesis perspective.
  
\section{Implementation}
\label{sec:implementation}


\lr is composed of
  approximately 13K lines of Racket
  plus approximately 58K lines
  of Racket
  automatically generated from 
  vendor-supplied Verilog.
Vendor-supplied Verilog
  was obtained from Lattice Diamond,
  Intel Quartus,
  and Xilinx Vivado
  sources.
We used Vivado version v2023.1,
  Quartus 22.1std.1 Build 917 02/14/2023 SC Lite Edition,
  Diamond version 3.12,
  Yosys version 0.36+42 (commit \texttt{70d3531}),
  the cvc5~\cite{barbosa22cvc5} and Yices2~\cite{dutertre2006yices,dutertre2014yices}
  solvers
  included in the 2023-08-06 release of \texttt{oss-cad-suite} from YosysHQ,
  the Bitwuzla solver at commit \texttt{b655bc0}~\cite{bitwuzla},
  the STP solver at commit \texttt{0510509a},
  Racket version 8.9~\cite{racket,racket:ref},
  and Rosette version 4.1~\cite{rosette4}.

\subsection{Primitive Interfaces}
\label{sec:impl-primitive-interfaces}

As described in \cref{sec:overview},
\textit{primitive interfaces}
  describe abstract versions
  of common FPGA primitives,
  which allow sketch templates
  to be architecture-independent.
To date, \lr declares primitive interfaces for
  $n$-input LUTs, $w$-width carry chains, 
  $n$-input muxes,
  and DSPs with up to four data inputs and one clock input.
The next section includes a concrete example
  of \lr's LUT4 primitive interface.\tighten


\subsection{Architecture Descriptions}
\label{sec:impl-arch-descs}

\begin{figure}
\begin{minted}[fontsize=\footnotesize]{YAML}
implementations:
  - interface: { name: LUT, num_inputs: 4 }
    internal_data: { sram: 16 }
    modules:
      - module_name: frac_lut4
        filepath: SOFA/frac_lut4.v
        ports:
          - { name: in, direction: in, width: 4, 
              value: (concat I3 I2 I1 I0) }
          - { name: mode, direction: in,
              width: 1, value: (bv 0 1) }
          - { name: lut4_out, direction: out,
              width: 1 }
        parameters: [{ name: sram, value: sram }]
    outputs: { O: lut4_out }
    \end{minted}
    \caption{
SOFA architecture description.
  }
    \label{fig:sofa-architecture-description}
\end{figure}

\noindent As described in \cref{sec:overview},
  \textit{architecture descriptions}
  convey the information
  required to convert
  each instance of a primitive interface
  into the corresponding
  architecture-specific module,
  which occurs while converting
  sketch templates
  into sketches.
The architecture description is
  the only additional input that
  may be required from a user
  to support a new architecture;
  it is a one-time effort 
  that is reusable
  for any designs in an architecture.
Architecture descriptions
  are simply lists
  (provided as YAML files)
  of the primitive interfaces
  that an architecture implements,
  but, crucially,
  also include 
  architecture-specific
  port and parameter values
  in a map called \texttt{internal\_data}.
Values in this map
  become symbolic values
  solvable by the SMT solver.
Additional constraints can also be specified
  in the architecture description 
  to rule out invalid configurations
  and minimize the solver's search space.\tighten

As an example,
  \cref{fig:sofa-architecture-description}
  shows the architecture description
  for the SOFA~\cite{sofa}
  FPGA architecture.
The description contains
  a single primitive interface implementation, i.e., 
  LUT4.
\lr's LUT4 primitive interface
  standardizes the names of a LUT4's inputs and outputs,
  naming the inputs
  \texttt{I0} through \texttt{I3}
  and the output
  \texttt{O}.
The SOFA implementation 
  of the LUT4 primitive interface
  uses
  the SOFA-specific
  \texttt{frac\_lut4}
  primitive.
Primitive interface inputs \texttt{I0} through \texttt{I3}
  are mapped to
  the actual
  input port of the \texttt{frac\_lut4},
  named \texttt{in}.
Likewise, the \texttt{frac\_lut4} output
  \texttt{lut4\_out}
  is mapped to the primitive interface output \texttt{O}.
The \texttt{internal_data} field
  declares \texttt{sram},
  the LUT's 16-bit internal memory,
  as an architecture-specific detail
  to be solved during synthesis.
 
If a sketch template uses a primitive interface
  not included in the architecture description
  (e.g., SOFA does not implement carries),
  \lr may still be able to implement the primitive interface
  based on primitive interfaces
  the architecture \textit{does} implement.
To date, \lr can implement any mux
  with LUTs, 
  a larger LUT from  smaller LUTs,
  a smaller LUT from a larger LUT,
  a carry from LUTs,
  and a smaller DSP from a larger DSP;
  it handles these conversions during sketch generation.

\subsection{Sketch Templates, Sketches, and Sketch Generation}

\label{sec:impl-sketch-templates}

As described in \cref{sec:overview},
  \lr 
  captures common FPGA implementation patterns
  in reusable, architecture-independent
  \textit{sketch templates.}
Thus far, we have described only 
  the relatively simple
  \texttt{dsp} sketch template,
  which instantiates a DSP.
As a more complex example
  of capturing common FPGA implementation patterns,
  consider
  the \texttt{bitwise-with-carry}
  sketch template, which
  uses $n$ LUTs and a carry chain
  to implement designs such as 
  addition or subtraction.
As of the paper's publication date, 
  \lr provides 5 sketch templates: 
  \texttt{dsp},
  \texttt{bitwise}, \texttt{bitwise-with-carry},
  \texttt{comparison} (LUT- and carry-based arithmetic comparison), 
  and
  \texttt{multiplication} (LUT-based multiplication).
 
The process of converting sketch templates
  to sketches
  is implemented as described in
  \cref{sec:overview}
  and \cref{sec:impl-arch-descs}.
\lr iterates over every
  primitive interface instance
  in the sketch
  and replaces it with
  the concrete primitive
  in accordance with
  the architecture's 
  architecture description.
If the architecture description
  does not implement the requested
  primitive interface,
  \lr checks whether it can implement
  the primitive interface with other
  implemented interfaces
  (e.g., implementing a smaller LUT with
  a larger LUT)
  and raises an error otherwise.

Sketch templates and sketches alike 
  are written in a domain-specific language (DSL)
  embedded into Rosette,
  whose implementation closely mirrors the syntax
  and semantics of \UberLang. 
The only significant difference is that
  the interpreter implementation
  does not use bitvector streams natively.
Instead, each invocation of the interpreter
  represents a single timestep,
  and all intermediate values from the previous timestep
  are taken as input.
Streams are then built up
  using multiple invocations of the interpreter.

\subsection{Importing Semantics from Verilog Modules}
\label{sec:implementation:importing-semantics}

\lr uses 
  Yosys~\cite{wolf2013yosys}
  to convert Verilog modules
  into the \texttt{btor2} format~\cite{btor} 
  and then converts the resulting \texttt{btor2}
  to Rosette/Racket code.

Due to the semantics of the Verilog language
  and the internal implementation of Yosys,
  extracting semantics from Verilog modules
  may require the following manual modifications
  to accommodate semantics extraction and synthesis:
  
\begin{itemize}[leftmargin=*]
\item As Yosys converts
  parameters from variables
  to constant values
  immediately upon module import,
  module parameters should be converted to
  ports
  to ensure they remain variables
  (and thus solvable by the SMT solver).
Note that not all parameters 
  can always be converted to ports,
  meaning some parameters cannot be solved for.
\item Strings should be converted to bitvectors.
\item All registers should be initialized.
\item All instances of \texttt{x} and \texttt{z} values should be  
  converted to 2-state logic (0 or 1).
\end{itemize}
Note that these caveats
  apply only 
  to our prototype implementation,
  not the general technique
  of semantics extraction from HDL.
Once these manual modifications are made,
  the following series of Yosys passes
  can be used to convert the Verilog
  into suitable \texttt{btor2}:
\texttt{prep; flatten; pmuxtree; opt_muxtree; clk2fflogic; prep; write_btor}.

We implement
  the translation from \texttt{btor2}
  to Rosette bitvector expressions
  as a 1:1 translation 
  since both languages
  are simply operations over bitvectors.

\subsection{Program Synthesis and Compilation to Verilog}
\label{sec:implementation:program-synthesis}

We implement
  the synthesis procedure
  defined in \cref{subsec:lr-correctness-and-completeness}
  with Rosette.
Multiple clock cycle guarantees,
  as described in \cref{subsec:lrfn-bmc},
  are implemented simply by making $c+1$
  total assertions,
  asserting the output of the input design
  and the sketch are equal
  after each of the $c+1$ timesteps.
We use a portfolio solving
  method, running
  Bitwuzla~\cite{niemetz2020bitwuzla},
  cvc5~\cite{barbosa22cvc5},
  Yices2~\cite{dutertre2006yices,dutertre2014yices},
  and STP~\cite{stp}
  in parallel
  and using results from the first
  solver to terminate.
To produce Verilog,
  \lr compiles the program from its internal DSL
  to the JSON format defined
  by Yosys using a straightforward translation 
  and then uses Yosys to output Verilog.

\subsection{Integration with Other Tools}

This paper describes \lr
  as a standalone tool,
  but the core \lr implementation
  could be integrated directly into
  existing tools.
Though out of scope for this paper,
  we have early,
  encouraging results
  integrating \lr
  as a Yosys pass
  that lets users tag modules
  with annotations similar to 
  (and much richer than) Xilinx's
  \texttt{use\_dsp} annotation.
We then map annotated modules 
  to primitives using \lr,
  which let us easily apply \lr to
  many fragments within a larger design.
We plan to more fully
  integrate \lr into Yosys in future work,
  which should radically improve the completeness
  of Yosys's DSP mapping ability,
  as shown in \cref{fig:xilinx-completeness}.\looseness-1

\section{Evaluation}
\label{sec:evaluation}

We now evaluate
  \lr in terms of completeness
  and extensibility.
In the following experiments,
  we target four FPGA architectures:
  \textbf{Xilinx UltraScale+},
  commonly used
  for large, high-performance workloads;
  \textbf{Lattice ECP5},
  commonly used 
  in low-power, low-cost scenarios;
  \textbf{Intel Cyclone 10 LP},
  an FPGA designed for low-cost,
  high-volume use cases,
  and 
  \textbf{SOFA}~\cite{sofa},
  a recent, open-source
  FPGA developed by the
  research community.
We compare \lr to existing
  technology mappers.
For
  Xilinx Ultrascale+, Lattice ECP5, and
  Intel Cyclone 10 LP,
  we compare \lr against both
  the open source toolchain Yosys~\cite{wolf2013yosys}
  and the
  state-of-the-art,
  proprietary, closed source
  toolchains
  for each architecture.\footnote{
    Again, licensing restrictions 
    prevent our naming the specific 
    proprietary tools, but
    they are familiar, standard packages 
    used by many hardware designers.}\tighten
The experiments were conducted 
  on a system running Ubuntu 20.04.3 
  with an AMD EPYC 7702P 64-Core CPU.
The resident set size of a single \lr process
  did not exceed
  300MB while running our evaluation.
We use the software versions listed in \cref{sec:implementation}.

\subsection{\lr Completeness}
\label{sec:completeness}
The reliance of many technology mappers,
  including state-of-the-art tools,
  on hand-written patterns
  leads them to fail when attempting to map many workloads
  that \textit{should} be mapped to a single DSP.
In particular, 
  the process of partial
  design mapping 
  (illustrated in \cref{sec:overview}) 
  becomes
  a laborious endeavor because
  of this incompleteness:
  hardware designers
  hand-instantiate DSPs rather
  than rely on substandard
  automated tooling, repeating
  the work each time they identify a potential
  opportunity to use a DSP.
\lr's greater mapping completeness 
  significantly
  reduces the burden on hardware
  designers during partial
  design mapping and marks 
  the first step in automated
  mapping for full designs.
  We next evaluate how \lr's program synthesis approach
  enables it to achieve greater completeness
  for these program fragments.
  
\paragraph{\textnormal{\textit{\textbf{Evaluation Setup.}}}}

We highlight three particularly complex
  DSPs for the Xilinx Ultrascale+,
  Lattice ECP5,
  and Intel Cyclone 10 LP architectures: 
  the Xilinx DSP48E2,
  Lattice ALU54A/MULT18X18C 
  (a single DSP composed of 
  two primitives),
  and Intel cyclone10lp\_mac\_mult.
SOFA provides no DSP, and is not included
  in this part of the evaluation.
For each architecture's DSP,
  we enumerate a large subset 
  of the designs
  theoretically mappable
  to a single DSP 
  according to its
  configuration manual.
This microbenchmark set
  aims to capture
  the real-world designs
  which hardware designers
  would attempt to map
  to a platform's DSP.
For each architecture, we compare \lr
  to both the corresponding
  state-of-the-art toolchain for
  the architecture 
  as well as to Yosys.
For Xilinx Ultrascale+, 
  the DSP48E2 configuration manual details the
  structure of designs mappable
  to the primitive.
Our designs for Xilinx include 
  all permutations of the design form
  $((a \pm b) * c) \odot d$,
  where $\odot \in \{\&, |, \pm, \oplus \}$,
  as well as designs of the forms $(a * b)$ and $((a * b) \pm c)$.
We pipeline each of these
  workloads from zero to three stages 
  and use bitwidths from 8 to 18 bits.
For the DSP on Lattice, we similarly enumerate
  all designs of the form $(a * b) \odot c$, 
  where $\odot \in \{\&, |, \oplus, \pm\}$, and of the form $(a * b)$.
For each of these designs, 
  we use zero to two stages
  and bitwidths from 8 to 18 bits.
This results in
    1320 microbenchmarks for Xilinx UltraScale+,
    396 for Lattice ECP5,
    and 66 for Intel Cyclone 10 LP.
Though \lr's output is correct by construction,
  we further validate its output
  by simulating each \lr-compiled design
  over thousands of consecutive cycles
  using Verilator.\tighten
  
\paragraph{\textnormal{\textit{\textbf{Comparison to Existing Toolchains.}}}}
As demonstrated in Figure~\ref{fig:xilinx-completeness} (top),
  \lr maps $44\times$ more
  designs than Yosys
  and $2.1\times$ more designs
  than the proprietary,
  state-of-the-art toolchain on Xilinx Ultrascale+.
On Lattice ECP5,
  \lr maps $6.0\times$ more
  designs than Yosys 
  and $3.6\times$ more designs
  than the proprietary,
  state-of-the-art toolchain.
On Intel Cyclone 10 LP,
  \lr successfully maps all designs:
  $3\times$ more designs
  than the proprietary,
  state-of-the-art toolchain for Intel.
Yosys fails to map a single design
  on Intel.
State-of-the-art toolchains
  for all architectures fail
  to map more than half
  of the queried designs.
\lr times out on 
less than 20\% of designs.%
\footnote{We restricted Rosette
synthesis time to 
120 seconds, 40 seconds, and 20 seconds for
Xilinx, Lattice, and Intel
respectively, and
marked failure past that (though bitvector synthesis problems
are decidable).}
Note that \lr
  returns ``UNSAT'' on 
  approximately 260
  designs on UltraScale+, i.e., 
  \lr claims there is
  \textit{no} possible mapping
  to a DSP48E2 for the
  requested workload.
In all of these cases,
  both Xilinx SOTA
  and Yosys
  agree with \lr 
  and do not map the designs
  to a single DSP.
We conclude that 
  the set of designs we presented in 
  \textit{Evaluation Setup}
  must be overly broad;
  though the documentation implies
  that all of these designs are mappable
  to a single DSP,
  all three Xilinx synthesis tools 
  surveyed indicate that they are
  indeed not mappable.

For timing, we compared the mapping time for each
  of the tools
  and report the results
  in Figure~\ref{fig:xilinx-completeness} (bottom).
The wide ranges for \lr
  show that solver time for different
  program synthesis queries
  is highly variable.
This is explored more deeply in
  \cref{fig:lakeroad-synth-time},
  which shows that most synthesis queries
  terminate quickly,
  with a long tail of slower queries.
Note that the state-of-the-art technology
mapper for Ultrascale+ has a slow running time
  due to its long start-up process.\tighten

Regarding which solvers
  in the portfolio were
  most useful,
  of all terminating
  (success or UNSAT)
  \lr experiments,
  Bitwuzla was the first
  to complete
  for 671 of them,
  STP for 519,
  Yices2 for 464,
  and cvc5 for 64.

\lr's greater completeness
  directly translates into resource reduction.
On average, for each microbenchmark,
  \lr uses 3.9 fewer LEs
  (logic elements: LUTs, muxes, or carry chains)
  and 7.5 fewer registers than the Xilinx SOTA,
  7.2 fewer LEs/11.9 fewer registers than the Lattice SOTA,
  8.2 fewer LEs/14.3 fewer registers than the Intel SOTA,
  and 33.3 fewer LEs/11.4 fewer registers than Yosys. 
In the real world, the small modules
  captured by our microbenchmarks
  may be reused dozens
  if not hundreds of times
  across a large design.
Thus, the sizable resource
  reduction \lr provides on a single
  microbenchmark
  will be multiplied significantly
  for an entire design.\tighten

\paragraph{\textnormal{\textit{\textbf{Discussion.}}}}
Compared to Yosys,
  it is clear that
  \lr provides more complete support
  for
  programmable DSPs.
However, \lr's greater completeness
  over Yosys
  is perhaps not surprising since 
  Yosys is an open-source tool
  still under active development.
Part of the appeal
  of the Yosys toolchain
  is the diversity of backends
  it can target;
  these results show that, if incorporated
  into Yosys, \lr
  would further increase
  Yosys's flexibility and generality.
Perhaps most surprising
  is that \lr is more complete
  than
  specialized proprietary toolchains. 
Even the UNSAT results \lr produces 
  can be useful to designers 
  since they indicate
  potential flaws
  in the documentation or vendor-provided semantics.
In the context of a larger
  synthesis tool, \lr
  would provide stronger
  guarantees for mapping
  modules of larger designs.

\begin{figure}
    \centering
\includegraphics[width=0.44\textwidth]{assets/succeeded\_vs\_failed\_all.png}

\vspace{1em}
\footnotesize

    \begin{tabular}{|l|S[table-format=3.2]|S[table-format=3.2]|S[table-format=3.2]|}
    \hline
     Tool & {Median Time (s)} & \multicolumn{2}{c|}{Min / Max Time (s)} \\ \hline
    \multicolumn{4}{|c|}{\textbf{Xilinx}} \\ 
    \hline
         \lr & 14.99 & 2.99 & 127.70 \\  \hline
         SOTA Xilinx & 261.61 & 227.82 & 598.67 \\ \hline
         Yosys & 14.97 & 6.66 & 21.10 \\ \hline
    \multicolumn{4}{|c|}{\textbf{Lattice}} \\
    \hline
         \lr & 9.49 & 6.70 & 55.23 \\ \hline
         SOTA Lattice & 2.32 & 0.95 & 4.52 \\ \hline
         Yosys & 2.31 & 0.90 & 4.01 \\ \hline
    \multicolumn{4}{|c|}{\textbf{Intel}} \\
    \hline
         \lr & 2.92 & 2.12 & 4.13 \\ \hline
         SOTA Intel & 38.73 & 19.11 & 43.49 \\ \hline
         Yosys & 0.96 & 0.48 & 1.88 \\ \hline
    \end{tabular}


    \caption{
Results of 
  the completeness experiments
  described in \cref{sec:completeness},
  measuring the completeness 
  of technology mapping tools for DSPs on Xilinx UltraScale+ and Lattice ECP5,
  plus timing information.
A single bar in the bar chart
  communicates,
  for a given FPGA architecture
  and technology mapper,
  the proportion of the microbenchmarks
  that the given technology mapper could map to a single DSP.
In \lr's case, experiments can 
  either succeed 
  (\lr maps the microbenchmark
    to a single DSP),
  timeout,
  or return UNSAT.
For the other tools,
  experiments can either
  succeed
  or fail
  (i.e., the tool returns a mapping,
    but the mapping uses more than
    a single DSP).
There are a total of 
  1320 experiments/microbenchmarks for Xilinx,
  396 for Lattice,
  and 66 for Intel.
    }
    \label{fig:xilinx-completeness}
\end{figure}
\begin{figure}
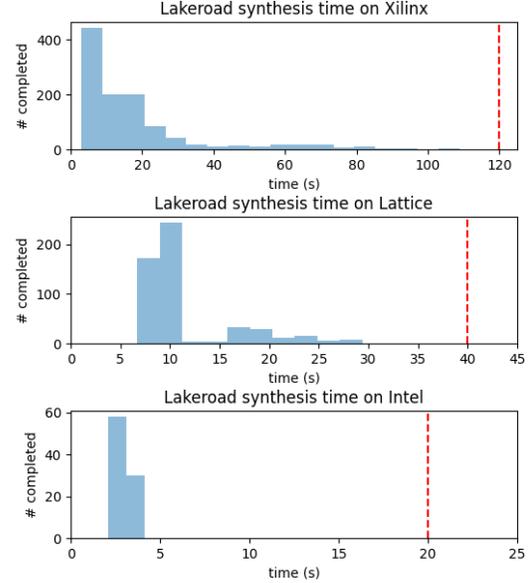

    \centering
    \includegraphics[width=0.43\textwidth]{assets/time/lakeroad\_time\_xilinx.png}
    \includegraphics[width=0.43\textwidth]{assets/time/lakeroad\_time\_lattice.png}
    \includegraphics[width=0.43\textwidth]{assets/time/lakeroad\_time\_intel.png}
\caption{
Histograms of \lr program synthesis runtime
  for all terminating
  (success or UNSAT) \lr experiments
  described in \cref{sec:completeness}, with timeout thresholds indicated with a vertical
  dotted red line.
}
    \label{fig:lakeroad-synth-time}
\end{figure}


\newcommand{\LatticeResourceTable}{add  & 1 & 1 & 2 &  &  & 1 &  &  &  \\
add  & 8 & 8 & 8 &  &  & 25 &  & 9 & 4 \\
add  & 32 & 32 & 32 &  &  & 135 &  & 55 & 28 \\
add  & 64 & 64 & 64 &  &  & 293 &  & 82 & 50 \\
and  & 1 & 1 &  &  &  & 1 &  &  &  \\
and  & 8 & 8 &  &  &  & 8 &  &  &  \\
and  & 32 & 32 &  &  &  & 32 &  &  &  \\
and  & 64 & 64 &  &  &  & 64 &  &  &  \\
not  & 1 & 1 &  &  &  & 1 &  &  &  \\
not  & 8 & 8 &  &  &  & 8 &  &  &  \\
not  & 32 & 32 &  &  &  & 32 &  &  &  \\
not  & 64 & 64 &  &  &  & 64 &  &  &  \\
mul  & 1 & 2 & 2 &  &  & 1 &  &  &  \\
mul  & 8 & 100 & 64 &  &  & 130 &  & 49 & 29 \\
mux  & 1 & 1 &  &  &  & 1 &  &  &  \\
mux  & 8 & 8 &  &  &  & 8 &  &  &  \\
mux  & 32 & 32 &  &  &  & 32 &  &  &  \\
mux  & 64 & 64 &  &  &  & 64 &  &  &  \\
sub  & 1 & 1 & 2 &  &  & 1 &  &  &  \\
sub  & 8 & 8 & 8 &  &  & 25 &  & 9 & 4 \\
sub  & 32 & 32 & 32 &  &  & 135 &  & 55 & 28 \\
sub  & 64 & 64 & 64 &  &  & 293 &  & 82 & 50 \\
}
\newcommand{\XilinxResourceTable}{add & 1 & 1 & 1 & \checkmark & 1 &  &  \\
add & 8 & 8 & 1 & \checkmark & 11 &  &  \\
add & 32 & 32 & 4 &  & 32 & 4 &  \\
add & 64 & 64 & 8 &  & 64 & 8 &  \\
add & 128 & 128 & 16 &  &  &  &  \\
and & 1 & 1 &  & \checkmark & 1 &  &  \\
and & 8 & 8 &  & \checkmark & 8 &  &  \\
and & 32 & 32 &  &  & 32 &  &  \\
and & 64 & 64 &  &  & 64 &  &  \\
and & 128 & 128 &  &  &  &  &  \\
not & 1 & 1 &  & \checkmark &  &  &  \\
not & 8 & 8 &  & \checkmark &  &  &  \\
not & 32 & 32 &  &  &  &  &  \\
not & 64 & 64 &  &  &  &  &  \\
not & 128 & 128 &  &  &  &  &  \\
mul & 1 & 2 & 1 & \checkmark & 1 &  &  \\
mul & 8 & 100 & 8 & \checkmark & 39 & 3 &  \\
mux & 1 & 1 &  & \checkmark & 1 &  &  \\
mux & 8 & 8 &  & \checkmark & 8 &  &  \\
mux & 32 & 32 &  &  & 32 &  &  \\
mux & 64 & 64 &  &  & 64 &  &  \\
mux & 128 & 128 &  &  &  &  &  \\
sub & 1 & 1 & 1 & \checkmark & 1 &  &  \\
sub & 8 & 8 & 1 & \checkmark & 11 &  &  \\
sub & 32 & 32 & 4 &  & 32 & 4 &  \\
sub & 64 & 64 & 8 &  & 64 & 8 &  \\
sub & 128 & 128 & 16 &  &  &  &  \\
}
\newcommand{\SofaResourceTable}{add  & 1 & 3 &  -  \\
add  & 8 & 24 &  -  \\
add  & 32 & 96 &  -  \\
add  & 64 & 192 &  -  \\
add  & 128 & 384 &  -  \\
and  & 1 & 1 &  -  \\
and  & 8 & 8 &  -  \\
and  & 32 & 32 &  -  \\
and  & 64 & 64 &  -  \\
and  & 128 & 128 &  -  \\
not  & 1 & 1 &  -  \\
not  & 8 & 8 &  -  \\
not  & 32 & 32 &  -  \\
not  & 64 & 64 &  -  \\
not  & 128 & 128 &  -  \\
mul  & 1 & 4 &  -  \\
mul  & 8 & 228 &  -  \\
mux  & 1 & 1 &  -  \\
mux  & 8 & 8 &  -  \\
mux  & 32 & 32 &  -  \\
mux  & 64 & 64 &  -  \\
mux  & 128 & 128 &  -  \\
sub  & 1 & 3 &  -  \\
sub  & 8 & 24 &  -  \\
sub  & 32 & 96 &  -  \\
sub  & 64 & 192 &  -  \\
sub  & 128 & 384 &  -  \\
}
\newcommand{\BigTable}{\rotatebox{\titletilt}{ Instr } & \rotatebox{\titletilt}{ Template } & \rotatebox{\titletilt}{ Runtime in sec } & \rotatebox{\titletilt}{ \#LUTs } & \rotatebox{\titletilt}{ \#CARRY8s } & \rotatebox{\titletilt}{ DSP? (sec) } & \rotatebox{\titletilt}{ Runtime } & \rotatebox{\titletilt}{ \#LUTs } & \rotatebox{\titletilt}{ \#MUXs } & \rotatebox{\titletilt}{ \#CARRY4s } & \rotatebox{\titletilt}{ \#INVs } & \rotatebox{\titletilt}{ DSP? } & \rotatebox{\titletilt}{ Runtime in sec } & \rotatebox{\titletilt}{ \#LUTs } & \rotatebox{\titletilt}{ \#CARRY8s } & \rotatebox{\titletilt}{ DSP? } & \rotatebox{\titletilt}{ Support } & \rotatebox{\titletilt}{ Runtime in sec } & \rotatebox{\titletilt}{ \#LUTs } & \rotatebox{\titletilt}{ \#CCU2Cs } & \rotatebox{\titletilt}{ DSP? } & \rotatebox{\titletilt}{ Runtime in sec } & \rotatebox{\titletilt}{ \#LUTs } & \rotatebox{\titletilt}{ \#MUXs } & \rotatebox{\titletilt}{ \#CCU2Cs } & \rotatebox{\titletilt}{ DSP? } & \rotatebox{\titletilt}{ Runtime in sec } & \rotatebox{\titletilt}{ \#LUTs } & \rotatebox{\titletilt}{ \#CCU2Cs } & \rotatebox{\titletilt}{ DSP? } & \rotatebox{\titletilt}{ Support } & \rotatebox{\titletilt}{ Runtime in sec } & \rotatebox{\titletilt}{ \#LUTs } \\\midrule
and8 & BW & 1.5 & 4 & -- & \checkmark (3.0) & 2.6 & 8 & -- & -- & -- &   & 32 & 4 & -- & ? & $\times$ & 1.6 & 8 & -- & ? & 1.5 & 8 & -- & -- &   & 1.1 & 8 & -- & ? & $\times$ & 1.8 & 8 \\
and16 & BW & 1.6 & 10 & -- & \checkmark (4.8) & 2.6 & 16 & -- & -- & -- &   & 36 & 10 & -- & ? & $\times$ & 1.6 & 16 & -- & ? & 1.5 & 16 & -- & -- &   & 0.6 & 16 & -- & ? & $\times$ & 2.2 & 16 \\
and32 & BW & 1.9 & 17 & -- &   & 2.6 & 32 & -- & -- & -- &   & 35 & 17 & -- & ? & $\times$ & 2.0 & 32 & -- & ? & 1.5 & 32 & -- & -- &   & 1.1 & 32 & -- & ? & $\times$ & 2.8 & 32 \\
or8 & BW & 1.7 & 4 & -- & \checkmark (7.4) & 2.5 & 8 & -- & -- & -- &   & 32 & 4 & -- & ? & $\times$ & 1.9 & 8 & -- & ? & 1.4 & 8 & -- & -- &   & 0.9 & 8 & -- & ? & $\times$ & 2.3 & 8 \\
or16 & BW & 1.8 & 10 & -- & \checkmark (7.1) & 2.6 & 16 & -- & -- & -- &   & 34 & 10 & -- & ? & $\times$ & 2.1 & 16 & -- & ? & 1.4 & 16 & -- & -- &   & 0.7 & 16 & -- & ? & $\times$ & 2.4 & 16 \\
or32 & BW & 1.7 & 17 & -- &   & 2.6 & 32 & -- & -- & -- &   & 31 & 17 & -- & ? & $\times$ & 2.2 & 32 & -- & ? & 1.5 & 32 & -- & -- &   & 0.6 & 32 & -- & ? & $\times$ & 3.1 & 32 \\
xor8 & BW & 1.7 & 4 & -- & \checkmark (5.1) & 2.6 & 8 & -- & -- & -- &   & 33 & 4 & -- & ? & $\times$ & 1.9 & 8 & -- & ? & 1.5 & 8 & -- & -- &   & 0.6 & 8 & -- & ? & $\times$ & 2.6 & 8 \\
xor16 & BW & 1.9 & 10 & -- & \checkmark (4.4) & 2.7 & 16 & -- & -- & -- &   & 32 & 10 & -- & ? & $\times$ & 2.0 & 16 & -- & ? & 1.5 & 16 & -- & -- &   & 1.0 & 16 & -- & ? & $\times$ & 2.4 & 16 \\
xor32 & BW & 2.2 & 17 & -- &   & 2.5 & 32 & -- & -- & -- &   & 33 & 17 & -- & ? & $\times$ & 1.8 & 32 & -- & ? & 1.5 & 32 & -- & -- &   & 0.7 & 32 & -- & ? & $\times$ & 2.4 & 32 \\
not8 & BW & 1.6 & 4 & -- & \checkmark (2.8) & 2.6 & -- & -- & -- & 8 &   & 34 & 4 & -- & ? & $\times$ & 1.8 & 8 & -- & ? & 1.4 & 8 & -- & -- &   & 0.7 & 8 & -- & ? & $\times$ & 1.9 & 8 \\
not16 & BW & 1.8 & 10 & -- & \checkmark (3.6) & 2.6 & -- & -- & -- & 16 &   & 32 & 10 & -- & ? & $\times$ & 1.8 & 16 & -- & ? & 1.5 & 16 & -- & -- &   & 0.6 & 16 & -- & ? & $\times$ & 2.6 & 16 \\
not32 & BW & 1.7 & 17 & -- &   & 2.6 & -- & -- & -- & 32 &   & 33 & 17 & -- & ? & $\times$ & 1.7 & 32 & -- & ? & 1.5 & 32 & -- & -- &   & 0.6 & 32 & -- & ? & $\times$ & 2.7 & 32 \\
mux8 & BW & 6.0 & 8 & -- & \checkmark (4.1) & 2.6 & 8 & -- & -- & -- &   & 35 & 8 & -- & ? & $\times$ & 2.0 & 8 & -- & ? & 1.6 & 8 & -- & -- &   & 1.1 & 8 & -- & ? & $\times$ & 2.7 & 8 \\
mux16 & BW & 8.5 & 16 & -- & \checkmark (6.3) & 2.6 & 16 & -- & -- & -- &   & 31 & 16 & -- & ? & $\times$ & 2.0 & 16 & -- & ? & 1.5 & 16 & -- & -- &   & 0.7 & 16 & -- & ? & $\times$ & 2.6 & 16 \\
mux32 & BW & 12.7 & 32 & -- &   & 2.8 & 32 & -- & -- & -- &   & 22 & 32 & -- & ? & $\times$ & 2.5 & 32 & -- & ? & 1.6 & 32 & -- & -- &   & 0.7 & 32 & -- & ? & $\times$ & 2.9 & 32 \\
add8 & Carry & 2.0 & 8 & 1 & \checkmark (3.7) & 2.6 & 8 & -- & 2 & -- &   & 33 & 11 & -- & ? & $\times$ & 2.3 & 8 & 4 & ? & 1.4 & -- & -- & 4 &   & 0.7 & 1 & 4 & ? & $\times$ & 2.6 & 24 \\
add16 & Carry & 2.3 & 16 & 2 & \checkmark (4.8) & 2.8 & 16 & -- & 4 & -- &   & 34 & 16 & 2 & ? & $\times$ & 2.4 & 16 & 8 & ? & 1.0 & -- & -- & 8 &   & 1.1 & 1 & 8 & ? & $\times$ & 2.9 & 48 \\
add32 & Carry & 2.3 & 32 & 4 &   & 2.7 & 32 & -- & 8 & -- &   & 34 & 32 & 4 & ? & $\times$ & 3.5 & 32 & 16 & ? & 1.5 & -- & -- & 16 &   & 0.7 & 1 & 16 & ? & $\times$ & 4.7 & 96 \\
sub8 & Carry & 2.0 & 8 & 1 & \checkmark (3.3) & 2.6 & 8 & -- & 2 & -- &   & 33 & 11 & -- & ? & $\times$ & 2.3 & 8 & 4 & ? & 1.4 & -- & -- & 4 &   & 0.8 & -- & 5 & ? & $\times$ & 2.6 & 24 \\
sub16 & Carry & 1.9 & 16 & 2 & \checkmark (3.9) & 2.7 & 16 & -- & 4 & -- &   & 34 & 16 & 2 & ? & $\times$ & 2.4 & 16 & 8 & ? & 1.4 & -- & -- & 8 &   & 0.8 & -- & 9 & ? & $\times$ & 3.2 & 48 \\
sub32 & Carry & 1.9 & 32 & 4 &   & 2.7 & 32 & -- & 8 & -- &   & 34 & 32 & 4 & ? & $\times$ & 3.3 & 32 & 16 & ? & 1.5 & -- & -- & 16 &   & 1.2 & -- & 17 & ? & $\times$ & 3.7 & 96 \\
mul8 & Mult & 2.0 & 91 & 8 & \checkmark (7.6) & 2.8 & 55 & 15 & 2 & -- &   & 36 & 33 & 3 & ? & $\times$ & 3.9 & 100 & 32 & ? & 1.9 & -- & -- & -- & \checkmark & 0.8 & -- & -- & ? & $\times$ & 6.2 & 228 \\
mul16 & DSP & 2.4 & -- & -- & \checkmark (2.4) & 3.4 & -- & -- & -- & -- & \checkmark & 32 & -- & -- & ? & $\times$ & -- & -- & -- & ? & -- & -- & -- & -- &   & -- & -- & -- & ? & $\times$ & -- & -- \\
mul32 & -- & -- & -- & -- &   & -- & -- & -- & -- & -- &   & -- & -- & -- & ? & $\times$ & -- & -- & -- & ? & -- & -- & -- & -- &   & -- & -- & -- & ? & $\times$ & -- & -- \\
eq8 & Comp & 1.8 & 8 & 1 &   & 2.6 & 3 & -- & -- & -- &   & 37 & 3 & -- & ? & $\times$ & 2.4 & 16 & 4 & ? & 1.4 & 5 & -- & -- &   & 0.8 & -- & 3 & ? & $\times$ & 4.6 & 32 \\
eq16 & Comp & 1.8 & 16 & 2 &   & 2.7 & 9 & 12 & -- & -- &   & 35 & 7 & -- & ? & $\times$ & 3.4 & 32 & 8 & ? & 1.5 & 12 & 1 & -- &   & 0.7 & -- & 5 & ? & $\times$ & 5.6 & 64 \\
eq32 & Comp & 3.3 & 32 & 4 &   & 2.6 & 24 & 1 & -- & -- &   & 36 & 11 & 2 & ? & $\times$ & 3.9 & 64 & 16 & ? & 1.6 & 25 & 2 & -- &   & 0.6 & -- & 9 & ? & $\times$ & 6.1 & 128 \\
neq8 & Comp & 2.1 & 8 & 1 &   & 2.6 & 3 & -- & -- & -- &   & 33 & 3 & -- & ? & $\times$ & 3.5 & 16 & 4 & ? & 1.5 & 5 & -- & -- &   & 0.7 & 1 & 3 & ? & $\times$ & 3.3 & 32 \\
neq16 & Comp & 2.6 & 16 & 2 &   & 2.7 & 9 & 12 & -- & -- &   & 23 & 7 & -- & ? & $\times$ & 2.6 & 32 & 8 & ? & 1.6 & 12 & 1 & -- &   & 1.1 & 1 & 5 & ? & $\times$ & 5.3 & 64 \\
neq32 & Comp & 3.5 & 32 & 4 &   & 2.7 & 24 & 1 & -- & -- &   & 33 & 11 & 2 & ? & $\times$ & 4.9 & 64 & 16 & ? & 1.6 & 25 & 2 & -- &   & 0.7 & 1 & 9 & ? & $\times$ & 7.0 & 128 \\
ugt8 & Comp & 2.7 & 8 & 1 &   & 2.6 & 6 & -- & 1 & -- &   & 37 & 4 & 1 & ? & $\times$ & 5.9 & 16 & 4 & ? & 1.5 & 1 & -- & 4 &   & 0.9 & 14 & -- & ? & $\times$ & 6.3 & 32 \\
ugt16 & Comp & 3.5 & 16 & 2 &   & 2.7 & 12 & -- & 2 & -- &   & 38 & 8 & 1 & ? & $\times$ & 7.4 & 32 & 8 & ? & 1.5 & 1 & -- & 8 &   & 1.0 & 35 & -- & ? & $\times$ & 6.1 & 64 \\
ugt32 & Comp & 3.4 & 32 & 4 &   & 2.7 & 22 & -- & 3 & -- &   & 32 & 16 & 2 & ? & $\times$ & 9.0 & 64 & 16 & ? & 1.5 & 1 & -- & 16 &   & 0.6 & 1 & 17 & ? & $\times$ & 10.0 & 128 \\
ult8 & Comp & 2.7 & 8 & 1 &   & 2.6 & 6 & -- & 1 & -- &   & 30 & 4 & 1 & ? & $\times$ & 5.8 & 16 & 4 & ? & 1.6 & 6 & 1 & 4 &   & 0.7 & 14 & -- & ? & $\times$ & 6.1 & 32 \\
ult16 & Comp & 3.3 & 16 & 2 &   & 2.7 & 12 & -- & 2 & -- &   & 32 & 8 & 1 & ? & $\times$ & 7.5 & 32 & 8 & ? & 1.6 & 16 & 2 & 8 &   & 1.0 & 35 & -- & ? & $\times$ & 6.2 & 64 \\
ult32 & Comp & 3.7 & 32 & 4 &   & 2.7 & 22 & -- & 3 & -- &   & 34 & 16 & 2 & ? & $\times$ & 9.1 & 64 & 16 & ? & 1.7 & 45 & 19 & 16 &   & 0.7 & 1 & 17 & ? & $\times$ & 10.3 & 128 \\
uge8 & Comp & 2.5 & 8 & 1 &   & 2.7 & 6 & -- & 1 & -- &   & 32 & 4 & 1 & ? & $\times$ & 4.6 & 16 & 4 & ? & 1.6 & 6 & 1 & 4 &   & 0.8 & 14 & -- & ? & $\times$ & 4.3 & 32 \\
uge16 & Comp & 2.4 & 16 & 2 &   & 2.7 & 12 & -- & 2 & -- &   & 34 & 8 & 1 & ? & $\times$ & 3.5 & 32 & 8 & ? & 1.6 & 16 & 2 & 8 &   & 1.0 & 35 & -- & ? & $\times$ & 6.7 & 64 \\
uge32 & Comp & 2.8 & 32 & 4 &   & 2.8 & 22 & -- & 3 & -- &   & 33 & 16 & 2 & ? & $\times$ & 11.4 & 64 & 16 & ? & 1.6 & 45 & 19 & 16 &   & 0.7 & -- & 17 & ? & $\times$ & 6.5 & 128 \\
ule8 & Comp & 2.6 & 8 & 1 &   & 2.7 & 6 & -- & 1 & -- &   & 34 & 4 & 1 & ? & $\times$ & 4.8 & 16 & 4 & ? & 1.4 & 6 & 1 & 4 &   & 0.6 & 14 & -- & ? & $\times$ & 3.8 & 32 \\
ule16 & Comp & 2.5 & 16 & 2 &   & 2.7 & 12 & -- & 2 & -- &   & 32 & 8 & 1 & ? & $\times$ & 3.6 & 32 & 8 & ? & 1.6 & 16 & 2 & 8 &   & 0.9 & 35 & -- & ? & $\times$ & 5.5 & 64 \\
ule32 & Comp & 2.5 & 32 & 4 &   & 2.8 & 22 & -- & 3 & -- &   & 32 & 16 & 2 & ? & $\times$ & 11.1 & 64 & 16 & ? & 1.8 & 45 & 19 & 16 &   & 0.9 & -- & 17 & ? & $\times$ & 5.7 & 128}
\newcommand{\LatticeTable}{\rotatebox{\titletilt}{ Instr } & \rotatebox{\titletilt}{ Template } & \rotatebox{\titletilt}{ Runtime in sec } & \rotatebox{\titletilt}{ \#LUTs } & \rotatebox{\titletilt}{ \#CCU2Cs } & \rotatebox{\titletilt}{ MUL? (sec) } & \rotatebox{\titletilt}{ Runtime in sec } & \rotatebox{\titletilt}{ \#LUTs } & \rotatebox{\titletilt}{ \#MUXs } & \rotatebox{\titletilt}{ \#CCU2Cs } & \rotatebox{\titletilt}{ \#MUL } & \rotatebox{\titletilt}{ Runtime in sec } & \rotatebox{\titletilt}{ \#LUTs } & \rotatebox{\titletilt}{ \#CCU2Cs } & \rotatebox{\titletilt}{ \#MUL } & \rotatebox{\titletilt}{ \#ALU } & \rotatebox{\titletilt}{ Support }  \\\midrule
and8 & BW    & 1.6  & 8 & --     &                   & 1.5 & 8 & -- & -- &         --      & 1.1 & 8  & -- &-- &-- & $\times$ \\
and16 & BW   & 1.6 & 16 & --    &                   & 1.5 & 16 & -- & -- &        --      & 0.6 & 16 & -- &-- &-- & $\times$ \\
and32 & BW   & 2.0 & 32 & --    &                   & 1.5 & 32 & -- & -- &        --      & 1.1 & 32 & -- &-- &-- & $\times$ \\
or8 & BW     & 1.9 & 8 & --       &                   & 1.4 & 8 & -- & -- &         --      & 0.9 & 8  & -- &-- &-- & $\times$ \\
or16 & BW    & 2.1 & 16 & --     &                   & 1.4 & 16 & -- & -- &        --      & 0.7 & 16 & -- &-- &-- & $\times$ \\
or32 & BW    & 2.2 & 32 & --     &                   & 1.5 & 32 & -- & -- &        --      & 0.6 & 32 & -- &-- &-- & $\times$ \\
xor8 & BW    & 1.9 & 8 & --      &                   & 1.5 & 8 & -- & -- &         --      & 0.6 & 8  & -- &-- &-- & $\times$ \\
xor16 & BW   & 2.0 & 16 & --    &                   & 1.5 & 16 & -- & -- &        --      & 1.0 & 16 & -- &-- &-- & $\times$ \\
xor32 & BW   & 1.8 & 32 & --    &                   & 1.5 & 32 & -- & -- &        --      & 0.7 & 32 & -- &-- &-- & $\times$ \\
not8 & BW    & 1.8 & 8 & --      &                   & 1.4 & 8 & -- & -- &         --      & 0.7 & 8  & -- &-- &-- & $\times$ \\
not16 & BW   & 1.8 & 16 & --    &                   & 1.5 & 16 & -- & -- &        --      & 0.6 & 16 & -- &-- &-- & $\times$ \\
not32 & BW   & 1.7 & 32 & --    &                   & 1.5 & 32 & -- & -- &        --      & 0.6 & 32 & -- &-- &-- & $\times$ \\
mux8 & BW    & 2.0 & 8 & --      &                   & 1.6 & 8 & -- & -- &         --      & 1.1 & 8  & -- &-- &-- & $\times$ \\
mux16 & BW   & 2.0 & 16 & --    &                   & 1.5 & 16 & -- & -- &        --      & 0.7 & 16 & -- &-- &-- & $\times$ \\
mux32 & BW   & 2.5 & 32 & --    &                   & 1.6 & 32 & -- & -- &        --      & 0.7 & 32 & -- &-- &-- & $\times$ \\
add8 & Carry & 2.3 & 8 & 4    &                   & 1.4 & -- & -- & 4 &         --      & 0.7 & 1  & 4  &-- &-- & $\times$ \\
add16 & Carry& 2.4 & 16 & 8  &                   & 1.0 & -- & -- & 8 &      --      & 1.1 & 1  & 8  &-- &-- & $\times$ \\
add32 & Carry& 3.5 & 32 & 16 &                   & 1.5 & -- & -- & 16 &     --      & 0.7 & 1  & 16 &-- &-- & $\times$ \\
sub8 & Carry & 2.3 & 8 & 4    &                   & 1.4 & -- & -- & 4 &      --      & 0.8  & --& 5  &-- &-- & $\times$ \\
sub16 & Carry& 2.4 & 16 & 8  &                   & 1.4 & -- & -- & 8 &      --      & 0.8  & --& 9  &-- &-- & $\times$ \\
sub32 & Carry& 3.3 & 32 & 16 &                   & 1.5 & -- & -- & 16 &     --      & 1.2 & -- & 17 &-- &-- & $\times$ \\
mul8 & Mult  & 3.9 & 100 & 32  & \checkmark (0.8)     & 1.9 & -- & -- & -- &     1       & 0.8 & -- & -- &-- &-- & $\times$ \\
mul16 & --  & -- & -- & --    & \checkmark (0.8)  & 1.3 & -- & -- & -- &      1      & 0.9  & -- & --  & 1 &-- & $\times$ \\
mul32 & --   & -- & -- & --     &                   & 1.1 & -- & -- & 14 &      3      & 2.8  & -- & --  & 4 & 2 & $\times$ \\
eq8 & Comp   & 2.4 & 16 & 4     &                   & 1.4 & 5 & -- & -- &      --      & 0.8 & -- & 3  &-- &-- & $\times$ \\
eq16 & Comp  & 3.4 & 32 & 8    &                   & 1.5 & 12 & 1 & -- &      --      & 0.7 & -- & 5  &-- &-- & $\times$ \\
eq32 & Comp  & 3.9 & 64 & 16   &                   & 1.6 & 25 & 2 & -- &      --      & 0.6 & -- & 9  &-- &-- & $\times$ \\
neq8 & Comp  & 3.5 & 16 & 4    &                   & 1.5 & 5 & -- & -- &      --      & 0.7 & 1  & 3  &-- &-- & $\times$ \\
neq16 & Comp & 2.6 & 32 & 8   &                   & 1.6 & 12 & 1 & -- &      --      & 1.1 & 1  & 5  &-- &-- & $\times$ \\
neq32 & Comp & 4.9 & 64 & 16  &                   & 1.6 & 25 & 2 & -- &      --      & 0.7 & 1  & 9  &-- &-- & $\times$ \\
ugt8 & Comp  & 5.9 & 16 & 4    &                   & 1.5 & 1 & -- & 4 &       --      & 0.9 & 14 & -- &-- &-- & $\times$ \\
ugt16 & Comp & 7.4 & 32 & 8   &                   & 1.5 & 1 & -- & 8 &       --      & 1.0 & 35 & -- &-- &-- & $\times$ \\
ugt32 & Comp & 9.0 & 64 & 16  &                   & 1.5 & 1 & -- & 16 &      --      & 0.6 & 1 & 17  &-- &-- & $\times$ \\
ult8 & Comp  & 5.8 & 16 & 4 &                      & 1.6 & 6 & 1 & 4 &           --      & 0.7 & 14 & -- &-- &-- & $\times$ \\
ult16 & Comp & 7.5 & 32 & 8 &                     & 1.6 & 16 & 2 & 8 &         --      & 1.0 & 35 & -- &-- &-- & $\times$ \\
ult32 & Comp & 9.1 & 64 & 16 &                    & 1.7 & 45 & 19 & 16 &      --      & 0.7 & 1 & 17  &-- &-- & $\times$ \\
uge8 & Comp  & 4.6 & 16 & 4 &                      & 1.6 & 6 & 1 & 4 &           --      & 0.8 & 14 & -- &-- &-- & $\times$ \\
uge16 & Comp & 3.5 & 32 & 8 &                     & 1.6 & 16 & 2 & 8 &         --      & 1.0 & 35 & -- &-- &-- & $\times$ \\
uge32 & Comp & 11.4 & 64 & 16 &   & 1.6 & 45 & 19 & 16 &     --      & 0.7 & -- & 17 &-- &-- & $\times$ \\
ule8 & Comp  & 4.8 & 16 & 4 &   & 1.4 & 6 & 1 & 4 &           --      & 0.6 & 14 & -- &-- &-- & $\times$ \\
ule16 & Comp & 3.6 & 32 & 8 &   & 1.6 & 16 & 2 & 8 &         --      & 0.9 & 35 & -- &-- &-- & $\times$ \\
ule32 & Comp & 11.1 & 64 & 16 &   & 1.8 & 45 & 19 & 16 &     --      & 0.9 & -- & 17 &-- &-- & $\times$}
\newcommand{\XilinxTable}{\rotatebox{\titletilt}{ Instr } & \rotatebox{\titletilt}{ Template } & \rotatebox{\titletilt}{ Runtime in sec } & \rotatebox{\titletilt}{ \#LUTs } & \rotatebox{\titletilt}{ \#CARRY8s } & \rotatebox{\titletilt}{ DSP? (sec) } & \rotatebox{\titletilt}{ Runtime } & \rotatebox{\titletilt}{ \#LUTs } & \rotatebox{\titletilt}{ \#MUXs } & \rotatebox{\titletilt}{ \#CARRY4s } & \rotatebox{\titletilt}{ \#INVs } & \rotatebox{\titletilt}{ \#DSP } & \rotatebox{\titletilt}{ Runtime in sec } & \rotatebox{\titletilt}{ \#LUTs } & \rotatebox{\titletilt}{ \#CARRY8s } & \rotatebox{\titletilt}{ \#DSP } & \rotatebox{\titletilt}{ Support } \\\midrule
and16 & BW    & 1.6  & 16 & -- & \checkmark (4.8) & 2.6 & 16 & -- & -- & -- &--   & 36 & 10 & -- &--   & $\times$ \\
and32 & BW    & 1.9  & 32 & -- &                  & 2.6 & 32 & -- & -- & -- &--   & 35 & 17 & -- &--   & $\times$ \\
or16  & BW    & 1.8  & 16 & -- & \checkmark (7.1) & 2.6 & 16 & -- & -- & -- &--   & 34 & 10 & -- &--   & $\times$ \\
or32  & BW    & 1.7  & 32 & -- &                  & 2.6 & 32 & -- & -- & -- &--   & 31 & 17 & -- &--   & $\times$ \\
xor16 & BW    & 1.9  & 16 & -- & \checkmark (4.4) & 2.7 & 16 & -- & -- & -- &--   & 32 & 10 & -- &--   & $\times$ \\
xor32 & BW    & 2.2  & 32 & -- &                  & 2.5 & 32 & -- & -- & -- &--   & 33 & 17 & -- &--   & $\times$ \\
not16 & BW    & 1.8  & 16 & -- & \checkmark (3.6) & 2.6 & -- & -- & -- & 16 &--   & 32 & 10 & -- &--   & $\times$ \\
not32 & BW    & 1.7  & 32 & -- &                  & 2.6 & -- & -- & -- & 32 &--   & 33 & 17 & -- &--   & $\times$ \\
mux16 & BW    & 8.5  & 16 & -- & \checkmark (6.3) & 2.6 & 16 & -- & -- & -- &--   & 31 & 16 & -- & -- & $\times$ \\
mux32 & BW    & 12.7 & 32 & -- &                  & 2.8 & 32 & -- & -- & -- &--   & 22 & 32 & -- & -- & $\times$ \\
add16 & Carry & 2.3  & 16 & 2 & \checkmark (4.8) & 2.8 & 16 & -- & 4 & -- &  -- & 34 & 16 & 2 & -- & $\times$ \\
add32 & Carry & 2.3  & 32 & 4 &                  & 2.7 & 32 & -- & 8 & -- &  -- & 34 & 32 & 4 & -- & $\times$ \\
sub16 & Carry & 1.9  & 16 & 2 & \checkmark (3.9) & 2.7 & 16 & -- & 4 & -- &  -- & 34 & 16 & 2 & -- & $\times$ \\
sub32 & Carry & 1.9  & 32 & 4 &                  & 2.7 & 32 & -- & 8 & -- &  -- & 34 & 32 & 4 & -- & $\times$ \\
mul16 & --   & 2.4  & -- & -- & \checkmark (2.4) & 3.4 & -- & -- & -- & -- & 1 & 32 & -- & -- & 1  & $\times$ \\
mul32 & --    & --   & -- & -- &                  & 3.6 & -- & -- & -- & -- & 3  & 22 & -- & -- & 3  & $\times$ \\
eq32  & Comp  & 3.3  & 32 & 4 &                   & 2.6 & 24 & 1 & -- & -- &  -- & 36 & 11 & 2 &--    & $\times$ \\
neq32 & Comp  & 3.5  & 32 & 4 &                   & 2.7 & 24 & 1 & -- & -- &  -- & 33 & 11 & 2 &--    & $\times$ \\
ugt32 & Comp  & 3.4  & 32 & 4 &                   & 2.7 & 22 & -- & 3 & -- &  -- & 32 & 16 & 2 &--    & $\times$ \\
ult32 & Comp  & 3.7  & 32 & 4 &                   & 2.7 & 22 & -- & 3 & -- &  -- & 34 & 16 & 2 &--   & $\times$ \\
uge32 & Comp  & 2.8  & 32 & 4 &                   & 2.8 & 22 & -- & 3 & -- &  -- & 33 & 16 & 2 &--   & $\times$ \\
ule32 & Comp  & 2.5  & 32 & 4 &                   & 2.8 & 22 & -- & 3 & -- &  -- & 32 & 16 & 2 &--  & $\times$}
\newcommand{\LatticeTable}{\rotatebox{\titletilt}{ Instr } & \rotatebox{\titletilt}{ Template } & \rotatebox{\titletilt}{ Runtime (s)} & \rotatebox{\titletilt}{ \#LUTs } & \rotatebox{\titletilt}{ \#CCU2Cs }  & \rotatebox{\titletilt}{ Runtime (s)} & \rotatebox{\titletilt}{ \#LUTs } & \rotatebox{\titletilt}{ \#MUXs } & \rotatebox{\titletilt}{ \#CCU2Cs } & \rotatebox{\titletilt}{ Runtime (s) } & \rotatebox{\titletilt}{ \#LUTs } & \rotatebox{\titletilt}{ \#CCU2Cs }  & \rotatebox{\titletilt}{ Support }  \\\midrule
and32 & BW   & 2.0 & 32 & --    & 1.5 & 32 & -- & -- &1.1 & 32 & -- & $\times$ \\
or32 & BW    & 2.2 & 32 & --    & 1.5 & 32 & -- & -- &0.6 & 32 & -- & $\times$ \\
xor32 & BW   & 1.8 & 32 & --    & 1.5 & 32 & -- & -- & 0.7 & 32 & -- & $\times$ \\
add32 & Carry& 3.5 & 32 & 16 & 1.5 & -- & -- & 16 & 0.7 & 1  & 16 & $\times$ \\
sub32 & Carry& 3.3 & 32 & 16 & 1.5 & -- & -- & 16 & 1.2 & -- & 17 & $\times$ \\
eq32 & Comp  & 3.9 & 64 & 16   & 1.6 & 25 & 2 & -- & 0.6 & -- & 9  & $\times$ \\
ugt32 & Comp & 9.0 & 64 & 16  & 1.5 & 1 & -- & 16 &  0.6 & 1 & 17  & $\times$ \\
ult32 & Comp & 9.1 & 64 & 16 & 1.7 & 45 & 19 & 16 &  0.7 & 1 & 17  & $\times$}
\newcommand{\SofaTable}{Instr & Template & Runtime in sec  & \#LUTs & Support & Support \\
\midrule
and8 & BW     & 1.8 & 8    & $\times$ & $\times$\\
and16 & BW    & 2.2 & 16   & $\times$ & $\times$\\
and32 & BW    & 2.8 & 32   & $\times$ & $\times$\\
or8 & BW      & 2.3 & 8    & $\times$ & $\times$\\
or16 & BW     & 2.4 & 16   & $\times$ & $\times$\\
or32 & BW     & 3.1 & 32   & $\times$ & $\times$\\
xor8 & BW     & 2.6 & 8    & $\times$ & $\times$\\
xor16 & BW    & 2.4 & 16   & $\times$ & $\times$\\
xor32 & BW    & 2.4 & 32   & $\times$ & $\times$\\
not8 & BW     & 1.9 & 8    & $\times$ & $\times$\\
not16 & BW    & 2.6 & 16   & $\times$ & $\times$\\
not32 & BW    & 2.7 & 32   & $\times$ & $\times$\\
mux8 & BW     & 2.7 & 8    & $\times$ & $\times$\\
mux16 & BW    & 2.6 & 16   & $\times$ & $\times$\\
mux32 & BW    & 2.9 & 32   & $\times$ & $\times$\\
add8 & Carry  & 2.6 & 24   & $\times$ & $\times$\\
add16 & Carry & 2.9 & 48   & $\times$ & $\times$\\
add32 & Carry & 4.7 & 96   & $\times$ & $\times$\\
sub8 & Carry  & 2.6 & 24   & $\times$ & $\times$\\
sub16 & Carry & 3.2 & 48   & $\times$ & $\times$\\
sub32 & Carry & 3.7 & 96   & $\times$ & $\times$\\
mul8 & Mult   & 6.2 & 228  & $\times$ & $\times$\\
mul16 & --   & -- & --    & $\times$ & $\times$\\
mul32 & --    & -- & --    & $\times$ & $\times$\\
eq8 & Comp    & 4.6 & 32   & $\times$ & $\times$\\
eq16 & Comp   & 5.6 & 64   & $\times$ & $\times$\\
eq32 & Comp   & 6.1 & 128  & $\times$ & $\times$\\
neq8 & Comp   & 3.3 & 32   & $\times$ & $\times$\\
neq16 & Comp  & 5.3 & 64   & $\times$ & $\times$\\
neq32 & Comp  & 7.0 & 128  & $\times$ & $\times$\\
ugt8 & Comp   & 6.3 & 32   & $\times$ & $\times$\\
ugt16 & Comp  & 6.1 & 64   & $\times$ & $\times$\\
ugt32 & Comp  & 10.0 & 128 & $\times$ & $\times$\\
ult8 & Comp   & 6.1 & 32   & $\times$ & $\times$\\
ult16 & Comp  & 6.2 & 64   & $\times$ & $\times$\\
ult32 & Comp  & 10.3 & 128 & $\times$ & $\times$\\
uge8 & Comp   & 4.3 & 32   & $\times$ & $\times$\\
uge16 & Comp  & 6.7 & 64   & $\times$ & $\times$\\
uge32 & Comp  & 6.5 & 128  & $\times$ & $\times$\\
ule8 & Comp   & 3.8 & 32   & $\times$ & $\times$\\
ule16 & Comp  & 5.5 & 64   & $\times$ & $\times$\\
ule32 & Comp  & 5.7 & 128  & $\times$ & $\times$   }
\newcommand{\SofaTable}{\rotatebox{\titletilt}{ Instr } & \rotatebox{\titletilt}{ Template } & \rotatebox{\titletilt}{ Runtime (s)} & \rotatebox{\titletilt}{ \#LUTs } & \rotatebox{\titletilt}{ Support } & \rotatebox{\titletilt}{ Support } \\\midrule
and32 & BW    & 2.8 & 32   & $\times$ & $\times$\\
or32 & BW     & 3.1 & 32   & $\times$ & $\times$\\
xor32 & BW    & 2.4 & 32   & $\times$ & $\times$\\
add32 & Carry & 4.7 & 96   & $\times$ & $\times$\\
sub32 & Carry & 3.7 & 96   & $\times$ & $\times$\\
eq32 & Comp   & 6.1 & 128  & $\times$ & $\times$\\
ugt32 & Comp  & 10.0 & 128 & $\times$ & $\times$\\
ult32 & Comp  & 10.3 & 128 & $\times$ & $\times$}
\newcommand{\FinalTable}{\rotatebox{\titletilt}{ Instr } & \rotatebox{\titletilt}{ Template } & \rotatebox{\titletilt}{ Runtime (s)} & \rotatebox{\titletilt}{ \#LUTs } & \rotatebox{\titletilt}{ \#CARRY8s } & \rotatebox{\titletilt}{ Runtime (s)} & \rotatebox{\titletilt}{ \#LUTs } & \rotatebox{\titletilt}{ \#MUXs } & \rotatebox{\titletilt}{ \#CARRY4s } & \rotatebox{\titletilt}{ \#INVs } & \rotatebox{\titletilt}{ Runtime (s)} & \rotatebox{\titletilt}{ \#LUTs } & \rotatebox{\titletilt}{ \#CARRY8s } & \rotatebox{\titletilt}{ Support } & \rotatebox{\titletilt}{ Runtime (s)} & \rotatebox{\titletilt}{ \#LUTs } & \rotatebox{\titletilt}{ \#CCU2Cs }  & \rotatebox{\titletilt}{ Runtime (s)} & \rotatebox{\titletilt}{ \#LUTs } & \rotatebox{\titletilt}{ \#MUXs } & \rotatebox{\titletilt}{ \#CCU2Cs } & \rotatebox{\titletilt}{ Runtime (s) } & \rotatebox{\titletilt}{ \#LUTs } & \rotatebox{\titletilt}{ \#CCU2Cs }  & \rotatebox{\titletilt}{ Support }  & \rotatebox{\titletilt}{ Runtime (s)} & \rotatebox{\titletilt}{ \#LUTs } & \rotatebox{\titletilt}{ Support } & \rotatebox{\titletilt}{ Support } \\\midrule
and32 & BW    & 1.9  & 32 & -- & 2.6 & 32 & -- & -- & -- & 35 & 17 & -- & $\times$ & 2.0 & 32 & --    & 1.5 & 32 & -- & -- &1.1 & 32 & -- & $\times$ & 2.8 & 32   & $\times$ & $\times$\\
or32  & BW    & 1.7  & 32 & -- & 2.6 & 32 & -- & -- & -- & 31 & 17 & -- & $\times$ & 2.2 & 32 & --    & 1.5 & 32 & -- & -- &0.6 & 32 & -- & $\times$ & 3.1 & 32   & $\times$ & $\times$\\
xor32 & BW    & 2.2  & 32 & -- & 2.5 & 32 & -- & -- & -- & 33 & 17 & -- & $\times$ & 1.8 & 32 & --    & 1.5 & 32 & -- & -- & 0.7 & 32 & -- & $\times$ & 2.4 & 32   & $\times$ & $\times$\\
add32 & Carry & 2.3  & 32 & 4  & 2.7 & 32 & -- & 8 & --  & 34 & 32 & 4 & $\times$  & 3.5 & 32 & 16 & 1.5 & -- & -- & 16 & 0.7 & 1  & 16 & $\times$ & 4.7 & 96   & $\times$ & $\times$\\
sub32 & Carry & 1.9  & 32 & 4 & 2.7 & 32 & -- & 8 & -- & 34 & 32 & 4 & $\times$ & 3.3 & 32 & 16 & 1.5 & -- & -- & 16 & 1.2 & -- & 17 & $\times$ & 3.7 & 96   & $\times$ & $\times$\\
eq32  & Comp  & 3.3  & 32 & 4 & 2.6 & 24 & 1 & -- & -- & 36 & 11 & 2 & $\times$ & 3.9 & 64 & 16   & 1.6 & 25 & 2 & -- & 0.6 & -- & 9  & $\times$  & 6.1 & 128  & $\times$ & $\times$\\
ugt32 & Comp  & 3.4  & 32 & 4 & 2.7 & 22 & -- & 3 & -- & 32 & 16 & 2 &$\times$ & 3.9 & 64 & 16   & 1.6 & 25 & 2 & -- & 0.6 & -- & 9  & $\times$ & 10.0 & 128 & $\times$ & $\times$\\
ult32 & Comp  & 3.7  & 32 & 4 & 2.7 & 22 & -- & 3 & -- & 34 & 16 & 2 &$\times$ & 9.1 & 64 & 16 & 1.7 & 45 & 19 & 16 &  0.7 & 1 & 17  & $\times$ & 10.3 & 128 & $\times$ & $\times$}
\newcommand{\titletilt}{90}

 


\subsection{\lr Extensibility and Expressiveness}

In addition to being
  correct by construction (\cref{sec:formalization}) and
  more complete 
  than existing FPGA technology mappers (\cref{sec:completeness}),
  \lr can also easily extend to new FPGA architectures.
Furthermore, automatic primitive semantics extraction 
  from vendor-provided HDL simulation models
  enables \lr to support diverse, highly configurable
  FPGA primitives.

The architecture descriptions
    vary in length from 20 to 240
    source lines of code (SLoC).
%
SOFA (20 SLoC) is the simplest, shown in full
  in \cref{fig:sofa-architecture-description}.
The descriptions for Xilinx (185 SLoC), 
  Lattice (240 SLoC), and Intel (178 SLoC)
  are longer since those
  FPGA architectures provide a
  wider range of configurable primitives.

As a point of comparison,
  the open-source Yosys toolchain,
  which has roughly 200 contributors on GitHub,
  provides technology mapping
  for Xilinx UltraScale+
  across over a dozen complex
  Verilog, C++, and Python files (about 1300
  lines of code).
We cannot provide similar numbers
  for state-of-the-art proprietary tools,
  but a developer
  of one such technology mapper
  shared that extending their tool to
  support new FPGA architectures
  was extremely difficult since it 
  ``spans millions of lines of low-level C.''
This is not surprising; Yosys aims to
  target a variety of vendor architectures, 
  while proprietary tools have teams of
  engineers to extract better mapping 
  (evident by Yosys' limitations
  in \cref{sec:completeness}).
By contrast,
  \lr supports
  diverse architectures and
  is easy to extend.
Even if a user
  wants to target a completely
  new architecture that
  \lr does not support,
  architecture-independent
  sketch templates allow reuse
  of previously implemented mapping
  strategies, and the user is
  only required to provide
  a few lines of
  high-level configuration
  for each primitive in 
  the architecture description.
  
\cref{table:imported-primitves} further
  highlights \lr's expressiveness,
  i.e., its ability to support a diverse
  range of configurable primitives
  by automatically extracting semantics from
  vendor-provided HDL simulation models.
\lr can import the semantics
  of large configurable primitives, 
  such as the UltraScale+ DSP (896 lines of Verilog)
  or Lattice ECP5's ALU and multiplier units (1642 and
  795 lines of Verilog, respectively).
It is difficult and error-prone
  to manually formalize the full semantics for these primitives;
  partial support by ad hoc search procedures
  that rely on syntactic pattern matching
  leads to missing many mapping opportunities,
  as shown in \cref{sec:completeness}.




\begin{table}
\caption{FPGA primitives imported 
  automatically by \lr from vendor-provided 
  Verilog models, with number of source lines of code 
  (excluding comments and empty lines) of 
  the original Verilog models.} 
\centering
\footnotesize
\label{table:imported-primitves}
\begin{tabular}{lrr}
 {\bf FPGA}   & \textbf{Primitive} & {\bf Verilog SLoC} \\\hline
 Xilinx Ultrascale+  & LUT6 & 88       \\
                     & CARRY8 & 23       \\
                     & DSP48E2 & 896       \\
                     \hline
 Lattice ECP5 & LUT2 & 5       \\
              & LUT4 & 7       \\
              & CCU2C & 60       \\
              & ALU54A & 1642 \\ 
              & MULT18X18C & 795 \\
              \hline
 Intel Cyclone 10 LP  & cyclone10lp\_mac\_mult   &  319       \\ 
              \hline
 SOFA      & frac\_lut4   &  69       \\ 
\end{tabular}
\end{table}

\section{Related Work}
\label{sec:background-and-related-work}

To the best of our knowledge,
  \lr is the first work
  to apply the technique of program synthesis
  to FPGA technology mapping.
Indeed, as noted by Sisco \textit{et al.}~\cite{sisco2022synthesis},
  program synthesis has seldom been applied
  in the domain of hardware design although
  its underlying formal methods techniques
  are frequently used for
  the \textit{formal verification}
  of hardware designs rather than compilation,
  as in Bluespec SystemVerilog~\cite{nikhil2004bluespec},
  \koika~\cite{bourgeat2020essence},
  and Kami~\cite{choi2017kami}.
Sisco \textit{et al.} cite two examples
  of works that use program synthesis for hardware design,
  Verisketch~\cite{ardeshiricham19verisketch} and Sketchilog~\cite{becker14sketchilog},
  both of which apply program synthesis to produce HDL implementations from high-level designs.
Other works use program synthesis
  to generate \textit{software}
  that runs on low-powered hardware,
  like Chlorophyll~\cite{phothilimthana2014chlorophyll},
  which targets extremely memory-constrained
  power-efficient processors,
  Chipmunk~\cite{gao2019chipmunk},
  which targets programmable network switches,
  and Diospyros~\cite{vanhattum2021vectorization},%
  \footnote{Diospyros uses symbolic evaluation, which is related to program synthesis, to lift imperative programs for digital signal processors into a high-level mathematical representation that can then be used with the technique of equality saturation~\cite{tate2011equality} to generate optimized code for the target devices. This is also distinct from the program synthesis techniques referenced elsewhere in this paper.}
  which generates vectorized programs for standalone digital signal processors (more powerful and general-purpose devices than the DSP units in FPGAs).
These works demonstrate the utility of program synthesis 
  for generating code that handles
  specific wrinkles in hardware designs,
  as does the use of program synthesis in \lr
  to harness the programmability of FPGA DSPs.

\lr is also related to past work in FPGA compilation and techmapping,
  much of which does not 
  entreaty to support
  programmable DSPs with as much generality.
ODIN~\cite{jamieson2005verilog} and ODIN-II~\cite{jamieson2010odin}
  are used in \textit{hard-block synthesis}
  for FPGAs,
  which is the task of mapping portions
  of hardware designs to specialized units (\textit{hard blocks}) like multipliers.
They operate purely over syntax (e.g., mapping \texttt{*} to a multiplier)
  and so are greatly limited in their ability
  to handle programmable DSPs.
The ABC~\cite{brayton2010abc} logic synthesis tool
  is used to lower hardware designs 
  into LUT and carry-chain configurations; it 
  is related to \lr in that it also uses constraint solvers
  to find configurations,
  though it is not general enough to handle
  a wide variety of programmable DSPs,
  unlike the program synthesis techniques used in \lr.
Note also that the use of
  configuration files in \lr to
  abstract away details of the FPGA architecture
  was inspired by past work in FPGA compilation,
  including OpenFPGA~\cite{tang2019openfpga}
  and the Verilog-to-Routing
  project (VTR)~\cite{rose2012vtr},
  both of which use abstract architecture descriptions
  to facilitate portability across designs,
  though these projects are limited in their support for DSPs.
Library-Parameterized Models~\cite{1993lpm, lpmaltera}
  define generic interfaces for common primitives and are also similar to \lr's primitive interfaces,
  though they are limited in their ability to represent configurable units like DSPs.\tighten

Virtual FPGA overlays~\cite{lysecky2005firm,brant2012zuma, landgraf2021compiler}
  are another approach to
  improving the mapping of hardware designs
  to hardware.
Overlays present a ``virtual''
  FPGA architecture;
  each actual architecture
  must then define a mapping
  from virtual to actual primitives.
This required translation is similar to
  \lr's requirement on users
  to implement primitive interfaces
  in an architecture description,
  though it requires more user effort.
The translation from virtual to actual architecture
  often comes with
  a steep resource
  and performance overhead.
\section{Conclusion}
\label{sec:conclusion}

This paper presents \lr, 
  a novel approach to FPGA technology mapping
  that leverages program synthesis techniques
  to provide stronger correctness and completeness guarantees 
  than state-of-the-art tools.
Because program synthesis tools
  can efficiently explore large search spaces, 
  \lr 
  can find mappings
  of hardware designs
  to FPGA DSPs
  in more cases
  than state-of-the-art tools,
  often finding more efficient implementations
  in the process.
With our techniques
  of semantics extraction
  from HDL
  and architecture-independent sketch templates,
  users must expend little manual effort 
  to apply \lr to
  a given FPGA architecture
  and extend it to handle further primitives.
Moreover, our formalization of \lr
  fosters greater confidence
  in its correctness.
\lr hence enables the 
  extensible, efficient, and correct 
  lowering of hardware designs to FPGAs,
  highlighting the effectiveness
  of program synthesis
  for FPGA technology mapping.


\section*{Acknowledgements}

This work was funded
  by generous grants
  and awards
  from
  Intel,
  the U.S. Department
  of Energy (award number DE-SC0022081),
  and the NSF (grant numbers 1836724 and 1749570).

We would like to thank
  our anonymous reviewers
  for their constructive feedback.
Thank you to
  Jonathan Balkind for serving
  as our shepherd.
Thank you to those who contributed code
  to early versions of \lr, including
  David Cao and Zihao Ye.
Thank you to Jin Yang and his
  team at Intel.
Thank you to 
  Daniel Petrisko, Scott Davidson,
  Rachit Nigam, and Adrian Sampson
  for sharing their deep knowledge of
  the hardware design workflow.
Thank you to Chandrakana Nandi
  for her enthusiasm and unwavering support.
Thank you to Claire Xenia Wolf, Nina Engelhardt,
  Jannis Harder,
  and the YosysHQ team.
Finally, thank you to the entire PLSE lab
  for their support and camaraderie.

\bibliographystyle{plain}
\bibliography{bib}

\begin{thebibliography}{10}
\bibitem{1993lpm}
Eia-is-103 : Library of parameterized modules (lpm).
\end{thebibliography}

\clearpage
\appendix
%
%
%
%
%

\section{Artifact Appendix}

\subsection{Abstract}

Our artifact
  consists of a zipfile
  containing the code
  for our evaluation.
Running the evaluation code
  will reproduce all of the figures
  present in this paper,
  which artifact evaluators can
  validate against our published data.
The evaluation code
  is comprised largely of
  the following files:
  documentation in a README,
  a Dockerfile
  to automatically set up the 
  evaluation environment,
  the \lr codebase, and
  the evaluation scripts themselves
  (a mix of Python and shell scripts).
The evaluation
  should be run on an x86 machine
  running Linux (ideally Ubuntu).
The evaluation benefits from many CPU cores.
The evaluation requires at least 300GB
  of free space,
  mostly for installing proprietary hardware
  toolchains.
\tighten

\subsection{Artifact check-list (meta-information)}

{\small
\begin{itemize}
  \item {\bf Algorithm: }
    Program synthesis via Rosette.
    Hardware synthesis via traditional hardware toolchains.
  \item {\bf Program: }
    \lr, the Rosette-based hardware synthesis tool
      presented in this paper,
      plus 
      Yosys, Xilinx Vivado, Lattice Diamond,
      and Intel Quartus,
      the baseline hardware synthesis tools
      we compare against.
  \item {\bf Run-time environment: }
    Linux, ideally Ubuntu.
  \item {\bf Hardware: }
    x86 CPU, ideally with many cores.
  \item {\bf Output: }
    Images and CSV files representing
      this paper's figures and tables.
  \item {\bf Experiments: }
    Each experiment is a single run
      of a hardware synthesis tool
      (either \lr or one of our baseline tools).
    The entire experiment consists of
      thousands of these tool runs.
  \item {\bf How much disk space required (approximately)?: }
    300GB.
  \item {\bf How much time is needed to prepare workflow (approximately)?: }
    4 hours:
      3 hours to set up proprietary hardware tools,
      1 hour to build Docker image.
  \item {\bf How much time is needed to complete experiments (approximately)?: }
    2 to 10+ hours, depending on the number of cores.
    On our 64-core machine,
      the evaluation takes
      about 4 hours.
  \item {\bf Publicly available?: }
    Yes, at \url{https://github.com/uwsampl/lakeroad-evaluation}
    and archived publicly on Zenodo, see 
    DOI link below.
    
  \item {\bf Code licenses (if publicly available)?: } 
    MIT.
  \item {\bf Workflow framework used?: }
    Python DoIt.
  \item {\bf Archived (provide DOI)?: }
  \url{https://doi.org/10.5281/zenodo.10515833}
\end{itemize}
}

\subsection{Description}

\subsubsection{How to access}

We recommend downloading the zipped code repository
  from the DOI link above.
The code can also be cloned from the GitHub
  repository linked above.

\subsubsection{Hardware dependencies}

x86 CPU, preferably with many cores.

\subsubsection{Software dependencies}

Linux-based OS, ideally Ubuntu.

\subsection{Installation}

Please refer to the README in the artifact.
A more readable version of the README can be viewed
  on the GitHub repository,
  or by converting the README
  using a tool like Pandoc.

\subsection{Experiment workflow}

Please refer to the README in the artifact.

\subsection{Evaluation and expected results}

Please refer to the README in the artifact.

\subsection{Methodology}

Submission, reviewing and badging methodology:

\begin{itemize}
  \item \url{https://www.acm.org/publications/policies/artifact-review-badging}
  \item \url{http://cTuning.org/ae/submission-20201122.html}
  \item \url{http://cTuning.org/ae/reviewing-20201122.html}
\end{itemize}

\end{document}